%% file: arxiv.tex
\documentclass{article}

\usepackage{amsthm}
\usepackage{amsmath}
\usepackage{amsfonts,amssymb}
\usepackage{graphicx,color}
\usepackage{boxedminipage}
\usepackage{framed}
\usepackage{mathtools}
\usepackage{enumitem}
\usepackage{thmtools}
\usepackage{thm-restate}
\usepackage{xspace}
\usepackage{authblk}
\usepackage{footnote}
\usepackage{nicefrac}
\usepackage[numbers]{natbib}
\usepackage[a4paper,top=2cm,bottom=2cm,left=1in,right=1in,marginparwidth=1.75cm]{geometry}
\usepackage{appendix}
\usepackage{todonotes}
\usepackage{footnote}
 \usepackage[pdftex, plainpages = false, pdfpagelabels, 
                 bookmarks=false,
                 bookmarksopen = true,
                 bookmarksnumbered = true,
                 breaklinks = true,
                 linktocpage,
                 pagebackref,
                 colorlinks = true,  
                 linkcolor = blue,
                 urlcolor  = blue,
                 citecolor = red,
                 anchorcolor = green,
                 hyperindex = true,
                 hyperfigures
                 ]{hyperref} 
 \usepackage{xifthen}
 \usepackage{tabularx}
\usepackage{tikz} 
\usetikzlibrary{calc,math,patterns,shapes,backgrounds}
\tikzstyle{path} = [color=black,opacity=.30,line cap=round, line join=round, line width=10pt]

\usepackage[nameinlink]{cleveref}

\newlength{\RoundedBoxWidth}
\newsavebox{\GrayRoundedBox}
\newenvironment{GrayBox}[1]%
   {\setlength{\RoundedBoxWidth}{.93\textwidth}
    \def\boxheading{#1}
    \begin{lrbox}{\GrayRoundedBox}
       \begin{minipage}{\RoundedBoxWidth}}%
   {   \end{minipage}
    \end{lrbox}
    \begin{center}
    \begin{tikzpicture}%
       \node(Text)[draw=black!20,fill=white,rounded corners,%
             inner sep=2ex,text width=\RoundedBoxWidth]%
             {\usebox{\GrayRoundedBox}};
        \coordinate(x) at (current bounding box.north west);
        \node [draw=white,rectangle,inner sep=3pt,anchor=north west,fill=white] 
        at ($(x)+(6pt,.75em)$) {\boxheading};
    \end{tikzpicture}
    \end{center}}     

\newenvironment{defproblemx}[2][]{\noindent\ignorespaces%
                                \FrameSep=6pt%
                                \parindent=0pt%
                \vspace*{-1.5em}
                \ifthenelse{\isempty{#1}}{%
                  \begin{GrayBox}{\textsc{#2}}%
                }{%
                  \begin{GrayBox}{\textsc{#2}  parameterized by~{#1}}%
                }
                \begin{tabular*}{\textwidth}{@{\hspace{.1em}} >{\itshape} p{1.8cm} p{0.8\textwidth} @{}}%
            }{
                \end{tabular*}%
                \end{GrayBox}%
                \ignorespacesafterend
            }

\newcommand{\defproblema}[3]{
  \begin{defproblemx}{#1}
    Input:  & #2 \\
    Task: & #3
  \end{defproblemx}
}%

\newtheorem{theorem}{Theorem}
\newtheorem{lemma}{Lemma}[section]
\newtheorem{corollary}[lemma]{Corollary}
\newtheorem{proposition}[lemma]{Proposition}
\newtheorem{claim}[lemma]{Claim}
\newtheorem{observation}[lemma]{Observation}
\newtheorem{definition}[lemma]{Definition}
\newtheorem{redrule}{Reduction Rule}

\crefname{redrule}{Reduction Rule}{Reduction Rules}

\DeclareMathOperator{\operatorClassNP}{NP}
\newcommand{\classNP}{\ensuremath{\operatorClassNP}\xspace}

\DeclareMathOperator{\operatorClassCoNP}{coNP}
\newcommand{\classCoNP}{\ensuremath{\operatorClassCoNP}}
\DeclareMathOperator{\operatorClassFPT}{FPT\xspace}
\newcommand{\classFPT}{\ensuremath{\operatorClassFPT}\xspace}
\DeclareMathOperator{\operatorClassW}{W}
\newcommand{\classW}[1]{\ensuremath{\operatorClassW[#1]}}

\DeclareMathOperator{\operatorClassXP}{XP\xspace}
\newcommand{\classXP}{\ensuremath{\operatorClassXP}\xspace}

\newcommand{\cqed}{\ensuremath{\lhd}}
\makeatletter
\newenvironment{claimproof}{\par
  \pushQED{\cqed}%
  \normalfont \topsep6\p@\@plus6\p@\relax
  \trivlist
  \item\relax
  {\itshape
    Proof of the claim\@addpunct{.}}\hspace\labelsep\ignorespaces
}{%
  \hfill\popQED\endtrivlist\@endpefalse
}
\makeatother

\newcommand{\mcc}{\textsc{Multicolored Clique}\xspace}
\newcommand{\dsc}{\textsc{Shortest Cycle Packing}\xspace}
\newcommand{\dsce}{\textsc{Edge-Disjoint Shortest Cycle Packing}\xspace}
\newcommand{\msc}{\textsc{Min-Sum Cycle Packing}\xspace}
\newcommand{\msedc}{\textsc{Min-Sum Edge-Disjoint Cycle Packing}\xspace}
\newcommand{\mvc}{\textsc{Min-Vector Cycle Packing}\xspace}
\newcommand{\mvedc}{\textsc{Edge-Disjoint Min-Vector Cycle Packing}\xspace}
\newcommand{\dsf}{\textsc{Disjoint Shortest Factors}\xspace}
\newcommand{\mcp}{\textsc{Min-Cut Packing}\xspace}
\newcommand{\ind}{\textsc{Independent Set}\xspace}
\newcommand{\dsp}{\textsc{Disjoint Shortest Paths}\xspace}
 
\newcommand{\minsumdp}{\textsc{Min-Sum Disjoint Paths}\xspace}

\DeclareMathOperator{\poly}{poly}

\newcommand{\OO}{\mathcal{O}}
\newcommand{\chordcol}{\mathcal{R}}
\newcommand{\pathcol}{\mathcal{P}}

\newcommand{\cP}{\mathcal{P}}
\newcommand{\cC}{\mathcal{C}}
\newcommand{\cT}{\mathcal{T}}
\newcommand{\cS}{\mathcal{S}}
\newcommand{\cL}{\mathcal{L}}
\newcommand{\cU}{\mathcal{U}}

\newcommand{\bl}{\ensuremath{\square}}

\title{Packing Short Cycles\thanks{This work was supported by the Research Council of Norway under BWCA project (grant no. 314528), the Franco-Norwegian AURORA project (grant no. 349476), the UKRI EPSRC (grant EP/V044621/1), the Swarnajayanti Fellowship (grant DST/SJF/MSA-01/2017-18), and the ERC Horizon 2020 research and innovation programme (grant no. 819416). Email addresses: \texttt{matthias.bentert@uib.no, fedor.fomin@uib.no, petr.golovach@uib.no, tuko@di.ku.dk, william.lochet@gmail.com, f.panolan@leeds.ac.uk, r.maadapuzhi-sridharan@warwick.ac.uk, saket@imsc.res.in, kirillsimonov@gmail.com}}}

\renewcommand\footnotemark{}

\author[1]{Matthias Bentert}
\author[1]{Fedor V. Fomin}
\author[1]{Petr A. Golovach}
\author[2]{Tuukka Korhonen}
\author[3]{William Lochet}
\author[4]{Fahad Panolan}
\author[5]{M. S. Ramanujan}
\author[6]{Saket Saurabh}
\author[7]{Kirill Simonov}

\affil[1]{University of Bergen, Norway}
\affil[2]{University of Copenhagen, Denmark}
\affil[3]{CNRS, LIRMM, Universit\'{e} de Montpellier, France}
\affil[4]{School of Computer Science, University of Leeds, UK}
\affil[5]{University of Warwick, UK}
\affil[6]{Institute of Mathematical Sciences, HBNI, India}
\affil[7]{Hasso Plattner Institute, University of Potsdam, Germany}

\date{}

\begin{document}

\maketitle

\begin{abstract} Cycle packing is a fundamental problem in optimization, graph theory, and algorithms.
Motivated by recent advancements in finding vertex-disjoint paths between a specified set of vertices that either minimize the total length of the paths [Björklund, Husfeldt, ICALP 2014; Mari, Mukherjee, Pilipczuk, and Sankowski, SODA 2024] or request the paths to be shortest [Lochet, SODA 2021], we consider the following cycle packing problems: \msc and \dsc.

In \msc, we try to find, in a weighted undirected graph, \(k\) vertex-disjoint cycles of minimum total weight. Our first main result is an algorithm that, for any fixed \(k\), solves the problem in polynomial time. We complement this result by establishing the W[1]-hardness of \msc parameterized by \(k\). The same results hold for the version of the problem where the task is to find $k$ edge disjoint cycles.

Our second main result concerns \dsc, which is a special case of \msc that asks to find a packing of \(k\) shortest cycles in a graph. We prove this problem to be fixed-parameter tractable (FPT) when parameterized by \(k\) on weighted planar graphs. We also obtain a polynomial kernel for the edge-disjoint variant of the problem on planar graphs. Deciding whether \msc is \classFPT on planar graphs and whether \dsc is \classFPT on general graphs remain challenging open questions.
\end{abstract}

 \newpage
\section{Introduction}\label{sec:introduction}

We consider the following problem.
 
\defproblema{\msc}%
{A graph $G$ with a weight function $w\colon E(G)\rightarrow \mathbb{Z}_{>0}$, an integer $k\geq 1$, and $\ell\in\mathbb{Z}_{\geq 0}$.}%
{Decide whether there is a family (packing) of $k$ (vertex-)disjoint cycles whose total length is at most $\ell$.
}
We also consider the variant of the problem, called \msedc, where the task is to find a packing of edge-disjoint cycles of total length at most $\ell$.

\msc could be considered a ``relaxation'' of the notoriously difficult  \minsumdp. Let us remind that in the \minsumdp{} problem, we are 
    given a graph with a set of terminal pairs $(s_1,t_1),\dots, (s_k,t_k)$. The task is either to connect all terminal pairs $(s_i,t_i)$ by pairwise vertex-disjoint paths of minimum total length or to decide that there is no set of pairwise disjoint paths.  
  Of course, the existence of a polynomial-time algorithm solving \minsumdp{} for fixed $k$ would imply a polynomial-time algorithm solving \msc. Unfortunately, no such algorithm is known. 
    \citet{BH19}
give an algorithm with running time~$\OO(n^{11})$
for finding two disjoint $s_i$-$t_i$-paths of minimal total length in an $n$-vertex graph. For~$k>2$, the complexity of the problem is a long-standing open problem. Whether the problem is polynomial-time solvable for $k=3$ remains open even on planar graphs. The main reason for our interest in \msc was the lack of any progress on the complexity of  \minsumdp.

Our first theorem establishes the membership of  \msc and \msedc in complexity class \classXP when parameterized by $k$. More precisely, we show the following.

\begin{restatable}{theorem}{thmminsum}\label{thm:minsumxpalg}
\msc and \msedc can be solved in $n^{\OO(k^6)}$ time, where 
$n$ is the number of vertices in the graph.
\end{restatable}

Further, we show that \Cref{thm:minsumxpalg} can be generalized for the variant of the problem with individual constraints on the length of the cycles in a packing. Formally, in \mvc, we are given a graph $G$, a positive integer $k$, and $k$ positive integers $\ell_1,\ldots,\ell_k$. Then the task is to find $k$ (vertex) disjoint cycles $C_1,\ldots,C_k$  so that the length of $C_i$ is upper-bounded by $\ell_i$ for each $i\in\{1,\dots, k\}$. We prove that \mvc and the variant where we want the cycles to be edge-disjoint, called \mvedc, can also be solved in $n^{\OO(k^6)}$ time.

We complement \Cref{thm:minsumxpalg}  by a computational lower bound, showing that the existence of an \classFPT algorithm for \msc is unlikely.

\begin{restatable}{theorem}{thmminsumlb}\label{thm:minsum-lb}
 \msc is \classW{1}-hard in subcubic graphs with unit edge weights when parameterized by~$k$.
\end{restatable}
Because the lower bound is obtained for subcubic graphs, the result also holds for \msedc.

\medskip 

A special but still fascinating case of \minsumdp{} is the \dsp problem. Here, we want to connect terminal pairs by vertex-disjoint \emph{shortest} paths. Equivalently, this is the variant of \minsumdp, where we request to connect terminals by disjoint paths whose total length does not exceed 
$\sum_{1 \leq i \leq k} \text{dist}(s_i, t_i)$.

The ``relaxation'' of \dsp as a cycle packing is the variant of \msc where all the cycles in the packing should be shortest cycles, that is, $\ell=kg$ where $g$ is the girth of $G$. (Let us remember that the girth of a weighted graph is the minimum length of its cycles.) 

\defproblema{\dsc}%
{A graph $G$ with a weight function $w\colon E(G)\rightarrow \mathbb{Z}_{>0}$ and an integer $k\geq 1$.}%
{Decide whether there is a family (packing) of $k$ vertex-disjoint cycles, each of minimum length.
}
For example, if all edges of $G$ are of weight one  and the girth of $G$ is three, then \dsc becomes the \textsc{Triangle Packing} problem.

By \Cref{thm:minsumxpalg},  \dsc is solvable in polynomial time for a fixed number~$k$ of cycles. We do not know whether it is  \classW{1}-hard or \classFPT parameterized by $k$. Our next theorem establishes the fixed-parameter tractability of \dsc on planar graphs. 

\begin{restatable}{theorem}{thmmincyclesvd}\label{thm:mincycles-vd}
\dsc is solvable  in time $k^{\OO(k)}\cdot n^{\OO(1)}$   on planar graphs with $n$ vertices.
\end{restatable}

 We do not know whether   \dsc admits a polynomial kernel on planar graphs. However, the methods developed for the proof of  \Cref{thm:mincycles-vd} allow establishing a polynomial kernel for the \dsce problem,  the version of   \dsc where the cycles should be edge-disjoint.

\begin{restatable}{theorem}{thmmincyclesed}\label{thm:mincycles-ed}
\dsce on planar graphs admits a polynomial kernel such that the output graph has $\OO(k^2)$ vertices.
\end{restatable}

The kernelization algorithm also implies that \dsce can be solved in $k^{\OO(k)}\cdot n^{\OO(1)}$ time on planar graphs.

Due to duality in planar graphs, \dsce in planar graphs is equivalent to finding $k$ disjoint minimum edge cuts; here we assume that two edge cuts $(X_1,Y_1)$ and~$(X_2,Y_2)$ are assumed to be disjoint if $E(X_1,Y_1)\cap E(X_2,Y_2)=\emptyset$. Here, for each~$i\in [2]$, $(X_i,Y_i)$ partitions the vertex set and $E(X_i,Y_i)$ denotes the set of edges with exactly one endpoint in $X_i$. Thus, \Cref{thm:mincycles-ed} implies the existence of a polynomial kernel for the problem of 
packing edge-disjoint minimum cuts.  We observe that this problem parameterized by $k$ on general graphs is \classW{1}-hard even when restricted to graphs with unit edge weights.

\paragraph{Remarks on impact of the weight function.} Our algorithm in Theorem \ref{thm:minsumxpalg} is strongly polynomial for each $k$. This rules out standard modifications used for converting weighted problems to the unweighted case, for instance, by subdividing edges according to their weight which could be exponential in the graph size. In fact, we do not use the fact that the weights are integers and this result can easily be generalized to positive rationals (or even positive reals assuming the real RAM model). The same holds for Theorem \ref{thm:mincycles-vd}.

\subsection{Overview of the methods}

Here, we give a high-level overview over how we achieve the different results.
\subsubsection{\msc}
The algorithms for \msc, \msedc, \mvc, and \mvedc are based on a combinatorial result about the number of paths in graphs that are no-instances for the \msc problem.
In particular, we prove that if an edge-weighted graph $G$ does not contain a collection of $k$ pairwise vertex-disjoint cycles of total length at most~$\ell$, then the number of paths of length at most~$\ell$ in $G$ is at most $n^{\OO(k^5)}$ (\Cref{lem:mainlemminsumcyc}).
Given this lemma, obtaining $n^{\OO(k^6)}$-time algorithms for these problems is quite simple as shown in the next paragraph.

In the case of \msc, we start enumerating paths of length at most $\ell$, and if there are more than $n^{\Omega(k^5)}$ of them, then we know that it is a yes-instances.
Otherwise, from the set of all paths of length at most $\ell$, we obtain the set of all cycles of length at most $\ell$ and then enumerate all collections of ~$k$ of them in time $n^{\OO(k^6)}$.
To generalize this to \mvc, we assume that $\ell_1$ is the smallest of the length bounds, apply the same enumeration strategy with the parameters~$\ell_1$ and $k$, and then branch on the $n^{\OO(k^5)}$ enumerated cycles, to again obtain the $n^{\OO(k^6)}$ running time.
These algorithms generalize to the edge-disjoint case simply by applying the same combinatorial result about vertex-disjoint cycles, but inserting the additional check when branching that the cycles we choose are edge-disjoint.

Thus, the main technical ingredient of the proof of \Cref{thm:minsumxpalg} is the proof of this combinatorial result (\Cref{lem:mainlemminsumcyc}).
Let us sketch the ideas here.
First, consider the case of~$k=1$.
If $G$ does not contain a cycle of length at most $\ell$, then the number of paths of length at most $\ell/2$ in $G$ is at most $\binom{n}{2}$, because having two distinct paths of length at most $\ell/2$ between the same pair of vertices would imply the existence of a cycle of length at most $\ell$.
Then, we observe that for any integer $c \ge 1$ and length-bound $\ell_0$, if $G$ contains at most $N$ paths of length at most $\ell_0$, then $G$ contains at most $N^c$ paths of length $c \cdot \ell_0$, and therefore conclude that $G$ contains at most $\OO(n^4)$ paths of length at most $\ell$.

Let us now extend the above argument to the general case of $k>1$. 
First, let $\tilde{C}$ be a subgraph of~$G$ comprising a maximal collection of pairwise vertex-disjoint cycles each of length at most $\ell/k$.
By our assumption, $\tilde{C}$ consists of at most $k-1$ cycles, and the graph $G - V(\tilde{C})$ does not contain a cycle of length at most $\ell/k$.
Let us then consider a path $P$ of length at most $\ell$ in $G$.
By considering how the path~$P$ intersects with $\tilde{C}$, we divide it into {\em segments} of two types: maximal subpaths of $P$ that are internally vertex-disjoint with $V(\tilde{C})$, and maximal subpaths of $P$ that are contained in $\tilde{C}$.

We next argue that if the number of such segments is more than $\Omega(k^4)$, then we can in fact construct a family of $k$ vertex-disjoint cycles of total length at most $\ell$, contradicting the premise of the lemma. 
In order to do so, we show that if there is a cycle $C$ in $\tilde{C}$ such that $P$ ``enters and exits'' $C$ more than~$\Omega(k^3)$ times, then we can find 
the claimed collection of $k$ vertex-disjoint cycles in $P \cup C$.
The proof of this proceeds by ``cleaning up'' the interactions of $P$ and $C$ into one of three 
cases with the help of the Erd{\"o}s-Szekeres theorem, and then demonstrating a simple construction in each of the cases.

After proving that $P$ decomposes into $\OO(k^4)$ segments that are either internally vertex-disjoint with~$\tilde{C}$ or contained in $\tilde{C}$, it remains to show that there are only $n^{\OO(k)}$ possible choices for each segment and since $P$ was chosen arbitrarily, this would in turn allow us to bound the number of possibilities for $P$.
For subpaths internally vertex-disjoint with $\tilde{C}$ we use the property that~$V(G) - V(\tilde{C})$ does not contain a cycle of length at most $\ell/k$, implying that $G - V(\tilde{C})$ contains at most $n^{\OO(k)}$ paths of length at most~$\ell$.
For subpaths contained in $\tilde{C}$, it is simple to observe that there are at most $\OO(n^2)$ possible choices for them, as each path in $\tilde{C}$ is defined by the choice of two vertices and the choice of a side to travel inside a cycle.
This concludes the informal overview of the proof of \Cref{thm:minsumxpalg}.

\subsubsection{Cycle packing on planar graphs}
The proof of \Cref{thm:mincycles-vd} is based on constructing a laminar family of disjoint cycles representing all shortest cycles and decomposing this family into a tree. We use random separation on such trees to ``filter out'' the most ``promising'' parts and combine branching arguments with an \classFPT algorithm for computing \ind in map graphs. In more detail:

Let $G$ be a weighted planar graph $G$ of (weighted) girth $g$. Without loss of generality, we can assume that each vertex and edge of $G$ is included in a  shortest cycle. Otherwise, we can preprocess the graph in polynomial time and delete vertices and edges that do not appear in a shortest cycle. 
Consider a plane embedding of $G$.
Packing \emph{facial} cycles of $G$ is equivalent to computing an \ind in a map graph. Let us recall that a map graph is the intersection graph of connected and internally disjoint regions of the Euclidean plane \cite{ChenGP02}. Indeed, facial cycles of $G$ are vertex-disjoint if and only if the corresponding vertices in the map graph induced by $G$ are not adjacent. For computing an \ind in a map graph, Chen~\cite{Chen01a} gave an EPTAS; the approximation algorithm of Chen immediately implies an \classFPT algorithm for \ind on map graphs, and hence for packing shortest facial cycles. The first natural approach to try would be to reduce \dsc to \ind in map graphs. This approach fails because a solution may contain many non-facial cycles. The main challenge in obtaining the algorithm for \dsc is identifying such cycles. However, we will use Chen's algorithm as a subroutine for an appropriate base case.

In order to handle non-facial cycles, we construct a  rooted tree 
whose nodes are shortest cycles, called a \emph{Laminar Shortest Cycle Tree} (LSCT).  
Constructions with similar properties are used to approximate the maximum size of cycle packings for {uncrossable families} of cycles~\mbox{\cite{GoemansW98,schlomberg2023packing,schlomberg2024improved}}, as well as for finding the minimum cycle basis for a planar graph (see, e.g.,~\cite{HartvigsenM94}). However, our algorithm needs a new decomposition with specific properties tailored to the properties of shortest cycles.  

Suppose that $\cC$ is a family of vertex-disjoint shortest cycles of $G$. Then $\cC$ is a laminar family of cycles. Some of these cycles may be nested. This means that a rooted tree can represent $\cC$ naturally (see~\Cref{fig:tree}). We assume that the leaves of the tree are facial cycles in the embedding. 

\begin{figure}[t]
\centering
\scalebox{0.7}{
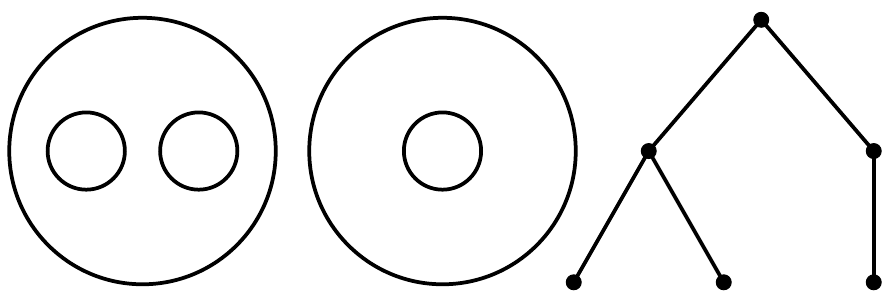}
\caption{A solution $\cC=\{C_1,\ldots,C_5\}$ and the tree representing $\cC$.}
\label{fig:tree}
\end{figure} 

It is convenient to assume that in the plane embedding of $G$ the frontier of the external face  is a shortest cycle. (It is not difficult to prove that such an embedding exists and it can be constructed in polynomial time.) 
  We fix such an embedding and construct an LSCT tree $\cT(G)$. The root of the tree is the facial cycle of the external face.   We construct the tree by recursively processing    shortest cycles that are current leaves of the already constructed tree.

If $C$  is a facial cycle, it becomes a leaf of $\cT(G)$. When  $C$  is not a facial cycle,  we consider two possibilities.

($i$) $C$ contains two vertices $s$ and $t$ (called the \emph{poles}) that are at a distance $g/2$ from each other in $C$ and such that there is an $s$-$t$-path $P$ of length $g/2$ whose edges and internal vertices are in the internal face of $C$. Then we say that $C$ is a \emph{splittable cycle} and we create an \emph{$S$-node} of $\cT(G)$ from $C$.
We find an inclusion-maximal family 
$\cP=\{P_0,P_1,\ldots,P_\ell\}$ of internally vertex-disjoint $s$-$t$-paths of length $g/2$ where~$P_0$ and~$P_\ell$ are the paths in $C$ and the other paths have their edges and internal vertices in the internal face of $C$.  
We assume the paths are indexed in the clockwise order from $s$.
We define the cycles formed by the consecutive paths in $\cP$ (see~\Cref{fig:decomp}(a)) to be the children of $C$. We show that the poles (if they exist) must be unique, so this step is well defined.

($ii$) If $C$ has no poles, we call such cycles  \emph{unsplittable} and we make a \emph{$U$-node} of $\cT(G)$ from~$C$. We find an inclusion-maximal laminar family 
$\cC=\{C_1\ldots,C_\ell\}$ of shortest cycles distinct from $C$ such that (i)~the cycles of $\cC$ are inside $C$ in the sense that they have no elements in the external face of $C$, (ii)~for every $i,j\in[\ell]$, cycles $C_i$ and $C_j$ are outside each other, that is, $C_j$ has no elements in the internal face of $C_i$ and vice versa, and (iii)~the cycles of $\cC$ are maximal in the sense that for each $i\in[\ell]$, there is no shortest cycle $C'$ distinct from $C$ and $C_i$ such that~$C_i$ is inside~$C'$. We set the cycles of $\cC$ to be the children of $C$ in $\cT(G)$ (see~\Cref{fig:decomp}(b)).  

We prove that an LSCT can be constructed in polynomial time and it contains all shortest cycles in the following sense: every shortest cycle $C$ is either a node of $\cT(G)$ or the tree has an $S$-node~$S$ with poles $s$ and $t$ such that $C$ is formed by two distinct paths of $\cP$ constructed for~$S$. 

The bottleneck in directly using LSCT at this point to find a solution, is dealing with $S$-nodes and cycles formed by paths of $\cP$ as these cycles are not nodes of the tree. To avoid having to consider such cycles in the solution, we argue that the random separation technique~\cite{CaiCC06} can be used to highlight pairs of paths in the families $\cP$ that could be cycles in a potential solution. We then modify $\cT(G)$ in such a way that for the obtained tree $\cT^*$, every cycle of some potential solution is a node of $\cT^*$. 

Finally,  we use a recursive branching algorithm to find a solution. First, we check in \classFPT time whether there is a solution formed by the facial cycles using the algorithm for \ind on map graphs.
Otherwise, if we fail to find a solution formed by facial cycles, we conclude that for a yes-instance, there should be a cycle $C$ in $\cT^*$ and a solution containing $C$ such that (i) the only cycles in this solution that are in the internal face of $C$ are facial cycles and (ii)~the solution contains $k'\geq 1$ facial cycles that are inside $C$. The crucial observation is that we can choose $C$ to be one of the lowest nodes in $\cT^*$ having the property that there are $k'$ vertex-disjoint facial cycles inside $C$, and we have at most~$k$ possibilities to choose $C$ for a given $k'$. Here, again we can use the algorithm for \ind on map graphs as a black box. We branch on $\OO(k^2)$ possible choices of $k'$ and $C$. For each branch, we delete the vertices and edges of $G$ that are in the inner face of $C$ and modify $\cT^*$ respectively,  decrease the parameter by $k'$, and recurse. This results in the algorithm solving
  \dsc in time~$k^{\OO(k)}\cdot n^{\OO(1)}$, which proves 
  \Cref{thm:mincycles-vd}.

  In the proof of \Cref{thm:mincycles-ed} which provides a polynomial kernel 
  for \dsce on planar graphs,  we use a similar approach. In this case the solution is  easier because when the LSCT has at least $4k$ leaves, that is, shortest facial cycles of $G$, the subgraph of the dual graph induced by the shortest facial cycles of $G$ has an independent set of size at least $k$ (by the Four-Color Theorem). This means that $(G,k)$ is a yes-instance of \dsce. In the case when the number of leaves is 
  less than $4k$, we can mark $\OO(k^2)$ nodes of $\cT(G)$ in such a way that a yes-instance has a solution containing only marked shortest cycles or cycles formed by paths from $\cP$ constructed for $S$-nodes. The idea behind the marking is to choose the lowest ``useful'' cycles. This leads to a kernel for \dsce with $\OO(k^2)$ vertices and edges and proves \Cref{thm:mincycles-ed}.

\subsection{Related Work}
The complexity of  \minsumdp{} is widely open on general and planar graphs.
 \citet{BH19} give a randomized algorithm with running time~$\OO(n^{11})$
for finding two disjoint $s_i$-$t_i$-paths of minimum total length. It is not known whether the problem is polynomial-time solvable for~$k=3$. 
On  planar graphs, several polynomial-time algorithms are known for the special cases when the terminals belong to two faces 
\cite{BorradaileNZ15,verdiere2011shortest,DattaIK018,Kobayashi22}. Recently, Mari et al.~\cite{MMS24} designed an algorithm for this problem with running time ~$k^{2^{\OO(k)}}$~when the input graph is a grid.

\dsp is \classNP-hard when~$k$ is part of the input \cite{Eil98} and \classW1-hard parameterized by~$k$~\cite{BNRZ23}.
\citet{Loc21} gives a polynomial-time algorithm for any fixed~$k$ in undirected graphs without weights, see also \citet{BNRZ23} for an improvement of the running time of Lochet's algorithm. Whether  \dsp is \classFPT parameterized by $k$ on planar graphs is an interesting open question.

Very little was known about the complexity of \msc and \dsc. Both problems are \classNP-complete when restricted to planar graphs~\cite{garey1979computers,GuruswamiRCCW98,Holyer81} because they generalize the well-known triangle packing problem. Rautenbach and Regen~\cite{RautenbachR09} show that for graphs of girth  $g \in \{4, 5\}$, \dsc and \dsce allow polynomial-time algorithms for instances with maximum degree $3$
 but are APX-hard for instances with maximum degree 4.  For each $g\geq 6$, both problems are APX-hard already for graphs with maximum degree $3$.
 
Approximation algorithms of ratio $(3+\varepsilon)$ for \dsc and \dsce on planar graphs follow from the framework of \citet{schlomberg2023packing} for uncrossing families of cycles.

When the sizes of the cycles are bounded, cycle packing is a special case of the subgraph packing problem. 
For such problems a plethora of parameterized and kernelization algorithms are known \cite{CyganFKLMPPS15}. However, all these algorithms work only for cycles of constant sizes. 
For vertex-disjoint (or edge-disjoint) cycle packing (without conditions on the lengths of cycles), a polynomial kernel on planar graphs, as well as more general classes of sparse graphs,  are known~\cite{BodlaenderFLPST16,FominLST20}.  

\section{Preliminaries}\label{sec:preliminaries}
For a positive integer $p$, we use $[p]$ to denote the set $\{1,\ldots,p\}$, $[p]_0$ for the set $\{0, 1,\ldots,p\}$ and we define~$[p,q]=\{p,\ldots,q\}$ for integers $p<q$.

\paragraph{Graphs.}
In this paper, we consider finite undirected graphs and refer to the textbook by Diestel~\cite{Diestel12} for undefined notion.
By default, the considered graphs are simple but we may allow multiple edges and loops in some occasions.
Let~$G=(V,E)$ be a graph. We use $V(G)$ and $E(G)$ to denote the set of vertices and the set of edges of $G$, respectively. We use~$n$ and~$m$ to denote the number of vertices and edges in~$G$, respectively, unless this creates confusion, in which case we clarify the notation explicitly.
For a vertex subset~$U \subseteq V$, we use~$G[U]$ to denote the subgraph of~$G$ induced by the vertices in~$U$ and~$G - U$ to denote~$G[V \setminus U]$.
For two graphs $G_1$ and $G_2$, the \emph{intersection}~$G=G_1\cap G_2$ is the graph with 
$V(G)=V(G_1)\cap V(G_2)$ and~$E(G)=E(G_1)\cap E(G_2)$, and the \emph{union} 
$G=G_1\cup G_2$ is the graph with 
$V(G)=V(G_1)\cup V(G_2)$ and $E(G)=E(G_1)\cup E(G_2)$.
A \emph{path}~$P=v_0\cdots v_\ell$ is a graph with vertex set~$\{v_0,v_1,\ldots,v_\ell\}$ and edge set~$\{\{v_{i-1},v_i\} \mid i \in [\ell]\}$; the vertices~$v_0$ and~$v_\ell$ are  the \emph{endpoints} of~$P$ and the other vertices are \emph{internal}. For a path with endpoints $s$ and $t$, we say that $P$ is an $s$-$t$-path. For a path $P$, a subgraph $R$ of $P$ with a single connected component is called a {\em subpath} of $P$ and we use the notation $R\subseteq P$ to denote this. 
In the case of unweighted graphs, the \emph{length} of $P$ is defined as the number of edges, and if we are given edge \emph{weights} $w\colon E(G)\rightarrow \mathbb{Z}_{>0}$ then the \emph{length} of~$P=v_0\cdots v_\ell$ is 
$w(P)=\sum_{i=1}^\ell w(\{v_{i-1},v_i\})$.
A \emph{cycle}~$C$ is a path along with an additional edge between the two endpoints.
The length of a cycle is defined in the same way as the length of a path. The \emph{girth} of $G$ is the minimum length of a cycle in $G$; $g(G)=+\infty$ if $G$ is a forest. An \emph{(edge) cut} of $G$ is a partition $(X,Y)$ of~$V(G)$; the set of edges $E(X,Y)=\{\{x,y\}\mid x\in X,~y\in Y\}$ is the \emph{cut-set}.  
If $G$ is a weighted graph then the weight of a cut is the weight of its cut-set.

A graph is \emph{planar} if it can be drawn on the plane such that its edges do not cross each other.
Such a drawing is called a \emph{planar embedding} of the graph and a planar graph with a planar embedding is called a \emph{plane} graph. The planarity  of a graph can be tested and a planar embedding can be found (if it exists) in linear time by the results of Hopcroft and Tarjan~\cite{HopcroftT74}. 
The \emph{faces} of a plane graph are the open regions of the plane bounded by a set of edges and that do not contain any other vertices or edges. The outer region is the \emph{external} face and the other faces are \emph{internal}. 
The vertices and edges 
appearing on the closed walk bounding the 
face form its \emph{frontier}.
Given a plane graph $G$, we use $F(G)$ to denote its set of faces.
If a cycle~$C$ of $G$ is the frontier of some face, then $C$ is a \emph{facial cycle}.
For a plane graph~$G$, its \emph{dual} graph~$G^*=(F(G),E^*)$ has the set of faces of $G$ as the vertex set, and for each $e\in E(G)$, $G^*$ has the \emph{dual} edge~$e^*$ whose endpoints are either two faces having~$e$ on their frontiers or~$e^*$ is a self-loop at~$f$ if~$e$ is in the frontier of exactly one face~$f$ (i.e., $e$ is a bridge of $G$). 
Observe that~$G^*$ is not necessarily simple even if $G$ is a simple graph. 
If $G$ is a weighted graph then it is standard to define the weights of the edges of the dual graph by setting the weight of $e^*$ to be equal to the weight of $e$ for each edge~$e\in E(G)$.
It is well known that finding a shortest cycle for a plane graph is equivalent to computing a minimum cut for the dual graph. More precisely, $C$ is a shortest cycle in a weighted plane graph $G$ if and only if the set $\{e^*\mid e\in E(C)\}$ of dual edges is a cut-set of~$G^*$ of minimum weight. In a weighted graph, the distance between two vertices is the weight of a minimum-weight path between them.

The celebrated four-color theorem~\cite{AppelH89,RobertsonSST97} implies the following observation about independent sets in planar graphs.

\begin{observation}\label{obs:ind-planar}
An $n$-vertex planar graph has an independent set of size at least $\frac{n}{4}$.
\end{observation}

A \emph{map} graph is the intersection graph of a finite family of simply connected and internally disjoint regions of the plane. In this paper, we assume that each map graph $G$ is given by its  embedding, that is, by a plane graph $H$ such that $V(G)\subseteq F(H)$ and two vertices $f_1,f_2\in V(G)$ are adjacent if and only if the frontiers of the faces $f_1$ and $f_2$ of $H$ share at least one point (a vertex or an edge of $H$). It was shown by Chen~\cite{Chen01a} that \ind is \classFPT on map graphs.\footnote{The result of Chen~\cite{Chen01a} is stated as an EPTAS but the approximation algorithm immediately implies an \classFPT algorithm.}

\begin{proposition}[\cite{Chen01a}]\label{prop:ind-map}
It can be decided in $2^{\OO(k)}\cdot|V(H)|^2$ time whether a map graph $G$ given by its embedding $H$ has an independent set of size at least $k$.
\end{proposition}

\paragraph{Parameterized Complexity.} We refer to the textbooks by Cygan et al.~\cite{CyganFKLMPPS15} and by Downey and Fellows~\cite{DowneyF13} for a detailed introduction. Here, we just briefly remind the main notions. A \emph{parameterized problem} is a language $L\subseteq\Sigma^*\times\mathbb{N}$  where $\Sigma^*$ is a set of strings over a finite alphabet $\Sigma$. An input of a parameterized problem is a pair $(x,k)$ where $x$ is a string over $\Sigma$ and $k\in \mathbb{N}$ is a \emph{parameter}. 
A parameterized problem is \emph{fixed-parameter tractable} (or \classFPT) if it can be solved in  $f(k)\cdot |x|^{\mathcal{O}(1)}$ time for some computable function~$f$.  
The complexity class \classFPT contains all fixed-parameter tractable parameterized problems.
A parameterized problem is in the class \classXP if it can be solved in $|x|^{f(k)}$ time for a computable function $f$.    
A \emph{kernelization algorithm} or \emph{kernel} for a parameterized problem $L$ is a polynomial-time algorithm that takes as its input an instance $(x,k)$ of $L$ and returns an instance $(x',k')$ of the same problem such that (i) $(x,k)\in L$ if and only if $(x',k')\in L$ and (ii) $|x'|+k'\leq f(k)$ for some computable function~$f\colon \mathbb{N}\rightarrow \mathbb{N}$. The function $f$ is the \emph{size} of the kernel; a kernel is \emph{polynomial} if $f$ is a polynomial. 

We conclude this section with two auxiliary results used in our paper. We need the following classical result of  Erd{\H{o}}s and Szekeres~\cite{erdos1935combinatorial}.  

\begin{proposition}[Special case of Erd\H{o}s-Szekeres theorem \cite{erdos1935combinatorial}]\label{prop:ErdosSzekeres}
Each sequence of more than ${(r-1)^2}$ distinct real numbers contains either an increasing subsequence of length $r$ or a decreasing subsequence of length $r$.
\end{proposition} 

For kernelization for \dsce on planar graphs, we use the algorithm of Frank and Tardos~\cite{FrankT87} to compress the weights as it is standard for kernelization algorithms~\cite{EtscheidKMR17}.  

\begin{proposition}[\cite{FrankT87}]\label{prop:FT}
There is an algorithm that, given a vector $w\in\mathbb{Q}^r$ and 
an integer~$N$, in (strongly) polynomial time finds a vector $\overline{w}\in\mathbb{Z}^r$ with $\|\overline{w}\|_{\infty}\leq 2^{4r^3}N^{r(r+2)}$ such that ${\mathsf{sign}(w\cdot b)=\mathsf{sign}(\overline{w}\cdot b)}$ 
for all vectors $b\in \mathbb{Z}^r$ with $\|b\|_1\leq N-1$.
\end{proposition}

\section{\msc is in XP}
In this section, we prove \Cref{thm:minsumxpalg} and show that \msc and \msedc can be solved in $n^{\OO(k^6)}$ time. The algorithms for these two problems are based on the same argument. We therefore focus on \msc and then explain how to extend the algorithm for the case of edge-disjoint cyles. 

First, 
we show that if a path and a cycle intersect in a complicated way, then we can construct a collection of many pairwise disjoint short cycles.
Before this, let us introduce some notation about the intersection of a path and a cycle.

Let $P$ be a path in $G$ and $C$ a cycle in $G$ so that $P$ and $C$ intersect in at least one vertex.
A \emph{chord} of $P$ relative to $C$ is a maximal non-empty subpath $R \subseteq P$  so that the endpoints of $R$ are in $V(C)$, the internal vertices of $R$ are disjoint from~$V(C)$, and~$E(R)$ is disjoint from~$E(C)$.
A \emph{tail} of $P$ relative to $C$ is a maximal non-empty subpath $T \subseteq P$ whose one endpoint is an endpoint of $P$ and is disjoint from~$V(C)$, the other endpoint is in $V(C)$, and all internal vertices are not contained in~$V(C)$.
Note that $P$ has between zero and two tails relative to~$C$.
Note also that $E(P)$ is the disjoint union of $E(P) \cap E(C)$, the edges of the chords of $P$ relative to~$C$, and the edges of the tails of $P$ relative to $C$.
\begin{observation}\label{obs:chordNumLikeComponentNum}
The number of chords of $P$ relative to $C$ is the number of connected components of~$P \cap C$ minus one.
\end{observation}

\begin{lemma}
\label{lem:chordbound}
Let $(G,w)$ be an edge-weighted graph, $C$ a cycle in $G$ of length at most $\ell$, and $P$ a path in $G$ of length at most $\ell$.
If $P$ has at least $128k^3$ chords relative to $C$, then $G$ contains a collection of $k$ pairwise vertex-disjoint cycles with total length at most $\ell$.
\end{lemma}
\begin{proof}
Let us assume that $P$ has at least $128k^3$ chords relative to $C$ and construct the desired collection of pairwise disjoint cycles.

First, we select a collection $\chordcol$ of at least $64k^3$ chords of $P$ so that the chords in $\chordcol$ are vertex-disjoint.
This is done simply by following along $P$ and selecting every second chord.
Now, let us index $V(C)$ with integers from $1$ to $|V(C)|$, in an order following the cycle.
Then we associate with each chord $R \in \chordcol$ the pair~$(s(R),t(R)) \in [|V(C)|]$ of numbers such that~$s(R) < t(R)$ correspond to the two endpoints of $R$ (which are in $V(C)$).
Note that as the chords are vertex disjoint, all of the numbers associated to them are distinct.

Let $R_1,R_2 \in \chordcol$ be two chords with $s(R_1) < s(R_2)$.
We say that $R_1$ and $R_2$ are 
\begin{itemize}
\item \emph{consecutive} if $s(R_1) < t(R_1) < s(R_2) < t(R_2)$,
\item \emph{crossing} if $s(R_1) < s(R_2) < t(R_1) < t(R_2)$, and
\item \emph{parallel} if $s(R_1) < s(R_2) < t(R_2) < t(R_1)$.
\end{itemize}
Note that any two chords are either consecutive, crossing, or parallel.

Let us then show that we can obtain a large enough subcollection of $\chordcol$ such that every pair of chords in this collection are related in the same way (i.e., one of consecutive, crossing, or parallel).

\begin{claim}
There is $\chordcol' \subseteq \chordcol$ with $|\chordcol'| \ge 4k$ so that all chords in $\chordcol'$ are either pairwise consecutive, pairwise crossing, or pairwise parallel.
\end{claim}
\begin{claimproof} We say $i \in [|V(C)|]$ \emph{touches} a chord $R \in \chordcol$ if $s(R) \le i \le t(R)$.
First, suppose that no $i \in [|V(C)|]$ touches more than $16k^2$ chords.
Then, the greedy algorithm that repeatedly chooses chords $R$ with the smallest~$t(R)$ and discards all other chords touched by~$t(R)$ manages to collect a set of at least $8k \geq 4k$ pairwise consecutive chords.

Otherwise, there is $i \in [|V(C)|]$ that touches at least $16k^2+1$ chords, implying that 
there is a collection $\chordcol_1 \subseteq \chordcol$ of at least $16k^2$ chords $R$ with $s(R) < i < t(R)$ (recall that the chords in $\chordcol$ are vertex disjoint). 
By the Erd\H{o}s-Szekeres theorem (\cref{prop:ErdosSzekeres}), 
there exists a subcollection $\chordcol_2 \subseteq \chordcol_1$ of size at least $4k$ so that when the chords in $\chordcol_2$ are sorted by $s(R)$, the endpoints $t(R)$ are either all in increasing order, or all in decreasing order.
In the former case, the chords in $\chordcol_2$ are pairwise crossing, and in the latter case, the chords in $\chordcol_2$ are pairwise parallel.
\end{claimproof}

We next consider the three cases arising from the above claim.
First, if we have a collection~$\chordcol'$ of $4k$ pairwise consecutive chords, then by forming for each chord $R \in \chordcol'$ a cycle consisting of~$R$ and the path in $C$ between $s(R)$ and $t(R)$, we obtain a collection of $4k$ pairwise vertex-disjoint cycles.
As the edges of these cycles come from~$E(P) \cup E(C)$, their total length is at most $2 \ell$, so by choosing the $k$ shortest of them we obtain a collection of~$k$ pairwise vertex-disjoint cycles with a total length at most $\ell/2$.

Second, suppose we have a collection $\chordcol'$ of $4k$ pairwise crossing chords.
Let us index them~$R_1, \ldots, R_{4k}$, so that $s(R_i) < s(R_j)$ if $i<j$.
Note that now, $s(R_{4k}) < t(R_1)$ and $t(R_i) < t(R_j)$ whenever $i<j$.
For each $i \in [2k]$, we construct a cycle $C_i$ by taking the union of~$R_{2i-1}$ and~$R_{2i}$, and connecting them by two paths in $C$, the first between~$s(R_{2i-1})$ and~$s(R_{2i})$, and the second between~$t(R_{2i-1})$ and~$t(R_{2i})$.
We obtain a collection of~$2k$ pairwise vertex-disjoint cycles.
The edges of these cycles are contained in~$E(P) \cup E(C)$, and therefore their total length is at most~$2\ell$.
By choosing the~$k$ shortest of them we obtain~$k$ pairwise disjoint cycles of total length of at most~$\ell$.

For the final case of $\chordcol'$ comprising~$4k$ pairwise parallel chords, 
let us use the same indexing~$R_1, \ldots, R_{4k}$ as in the previous paragraph, so that $s(R_i) < s(R_j)$ if $i<j$.
Note that now, $s(R_{4k}) < t(R_{4k})$ and~${t(R_i) > t(R_j)}$ whenever $i<j$.
We again construct a cycle~$C_i$ for each~$i \in [2k]$, by taking the union of~$R_{2i-1}$ and~$R_{2i}$, and connecting them by two paths in~$C$, the first between~$s(R_{2i-1})$ and~$s(R_{2i})$, and the second between~$t(R_{2i})$ and~$t(R_{2i-1})$.
We obtain a collection of~$2k$ pairwise vertex-disjoint cycles.
The edges of these cycles are contained in~$E(P) \cup E(C)$, and therefore their total length is at most~$2\ell$ and by choosing the~$k$ shortest of them we obtain~$k$ pairwise disjoint cycles with a total length of at most~$\ell$.
\end{proof}

Our algorithm for {\msc} is based on a graph-theoretical lemma that bounds the number of paths of length at most $\ell$ in no-instances.
We first prove this lemma in the case when~$k=1$, which will be used in the proof of the general case.

\begin{lemma}
\label{lem:girthpathbound}
If an edge-weighted $n$-vertex graph $(G,w)$ does not contain a cycle of length at most $\ell$, then the number of paths of length at most $\ell/2$ in $G$ is at most ${n \choose 2}$.
\end{lemma}
\begin{proof}
Suppose there are two vertices~$a,b \in V(G)$ with two distinct (but not necessarily vertex- or edge-disjoint)~$a$-$b$-paths $P_1$ and $P_2$ in $G$, both of length at most $\ell/2$.
Now, the sum of edge weights in~$P_1 \cup P_2$ is at most~$\ell$ and must contain a cycle, that is, $G$ contains a cycle of length at most~$\ell$.
Hence, for each pair $a,b \in V(G)$, there is at most one $a$-$b$-path of length at most $\ell/2$, and therefore the total number of paths of length at most $\ell/2$ is at most ${n \choose 2}$.
\end{proof}

We also note that such bounds on the number of paths of a certain length can be boosted to higher lengths.

\begin{observation}
\label{lem:pathboundboost}
Let $(G,w)$ be an edge-weighted graph, so that the number of paths of length at most~$\ell$ in $G$ is at most~$N$.
For any integer $c \ge 1$, the number of paths of length at most $c \cdot \ell$ in~$G$ is then at most $N^c$.
\end{observation}
\begin{proof}
Any path of length at most $c \cdot \ell$ can be constructed as the union of at most $c$ paths of length at most $\ell$.
\end{proof}

We are now ready to prove our main graph-theoretical lemma, which 
is at the heart of our algorithm.

\begin{lemma}
\label{lem:mainlemminsumcyc}
If an edge-weighted $n$-vertex graph $(G,w)$ does not contain a collection of $k$ pairwise vertex-disjoint cycles of total length at most $\ell$, then the number of paths of length at most $\ell$ in $G$ is in~$n^{\OO(k^5)}$.
\end{lemma}
\begin{proof}

Let $\tilde{C}$ be a subgraph of $G$ comprising a maximal collection of vertex-disjoint cycles each of length at most $\ell/k$.
Note that $\tilde{C}$ consists of at most $k-1$ cycles as otherwise~$G$ contains~$k$ vertex-disjoint cycles of total length at most~$\ell$.
Since  $G - V(\tilde{C})$ does not contain a cycle of length at most $\ell/k$, \Cref{lem:girthpathbound} implies that the number of paths of length at most~$\ell/(2 \cdot k)$ in~$G - V(\tilde{C})$ is at most ${n \choose 2} \leq n^2$.
Finally, \Cref{lem:pathboundboost} allows us to conclude that the number of paths of length at most $\ell$ in $G - V(\tilde{C})$ is at most~$n^{4k}$.
This also implies that the number of paths of length at most~$\ell$ in $G$ whose internal vertices are disjoint from $V(\tilde{C})$ is at most~$n^{8k}$.
Now, let $P$ be a path in $G$ of length at most $\ell$.

\begin{claim}
The number of connected components of $P \cap \tilde{C}$ is at most~$128k^4-1$.
\end{claim}
\begin{claimproof}
Suppose the number of connected components of $P \cap \tilde{C}$ is more than~${128k^3 \cdot (k-1)}$.
Since~$\tilde{C}$ is the union of at most~$k-1$~cycles of length at most $\ell/k$ each, by the pigeonhole principle, there is a cycle~$C \subseteq \tilde{C}$ such that the number of connected components of~$P \cap C$ is more than $128k^3$.
\Cref{obs:chordNumLikeComponentNum} now implies that~$P$ has at least $128k^3$ chords relative to $C$ and \Cref{lem:chordbound} therefore states that~$G$ contains a collection of $k$ pairwise vertex-disjoint cycles of total length at most~$\ell$, which is a contradiction.
\end{claimproof}

Now, $P$ can be constructed as the union of the subpaths of~$P$ that are internally vertex-disjoint from~$V(\tilde{C})$ and the connected components of~$P \cap \tilde{C}$, which are paths.
Notice that the number of subpaths of~$P$ that are internally disjoint with $V(\tilde{C})$
is by \Cref{obs:chordNumLikeComponentNum} at most the number of the connected components of~$P \cap \tilde{C}$ plus one, that is, at most~$128 k^4$.
On the other hand, the number of paths that are subgraphs of $\tilde{C}$ is at most $n^2$.
Hence, there are at most~$(n^2)^{128k^4}$ possible choices for subpaths of~$P$ that are connected components of~$P \cap \tilde{C}$ and for each of the at most~$128k^4$ chords and/or tails of~$P$ relative to~$\tilde{C}$, there are at most~$n^{8k}$ choices for paths which are internally vertex-disjoint from~$\tilde{C}$ as noted earlier.
The number of paths of length at most~$\ell$ in $G$ is therefore at most
\[(n^2)^{128k^4} \cdot (n^{8k})^{128k^4} = n^{256k^4} n^{1024k^5} \in n^{\OO(k^5)}.\qedhere\]
\end{proof}

Now, we are ready to translate \Cref{lem:mainlemminsumcyc} into an algorithm for \msc{}.
For this, we need the following easy lemma.

\begin{lemma}
\label{lem:pathgen}
There is an algorithm that, given an edge-weighted $n$-vertex graph $(G,w)$ and integers~$\ell$ and $N$, in time $N \cdot n^{\OO(1)}$ either returns all paths in $G$ of length at most $\ell$, or concludes that the number of such paths is more than $N$.
\end{lemma}
\begin{proof}
The set of all paths of length $0$ is easy to generate, as they are just the vertices of $G$.
Now, let~$i \ge 1$ and let $\pathcol_{i-1}$ be the set of paths of length at most $i-1$ in $G$.
Given $\pathcol_{i-1}$, we can generate in time $|\pathcol_{i-1}| \cdot n^{\OO(1)}$ the set $\pathcol_i$ of all paths of length at most $i$ by first generating~$|\pathcol_{i-1}| \cdot n^{\OO(1)}$ candidates by trying to extend each path by one edge (and also including all paths in $\pathcol_{i-1}$), then filtering out the obtained subgraphs that are not paths, and finally deduplicating the output by bucket sorting. 
By repeating this until either $i=\ell$ or $|\pathcol_i| > N$, we obtain the desired algorithm.
\end{proof}

We now prove the main theorem of the section. 

\thmminsum*

\begin{proof}
We demonstrate an algorithm for \msc and then explain how to adapt it for \msedc.

Let us first construct a decision algorithm, that given $(G,w,k,\ell)$, decides whether $G$ contains a collection of $k$ pairwise vertex-disjoint cycles of total length at most~$\ell$.
We first apply the algorithm of \Cref{lem:pathgen} with~$N =  n^{256k^4} n^{1024k^5} \in n^{\OO(k^5)}$ (as specified in the proof of~\Cref{lem:mainlemminsumcyc}).
If it concludes that the number of paths of length at most~$\ell$ is more than $N$, then we conclude that the answer is yes.
Otherwise, we obtain the collection of all paths of~$G$ of length at most~$\ell$, which has size at most~$N \in n^{\OO(k^5)}$.
By trying to extend each with an edge, we obtain the collection of size at most~$N$ of all cycles of~$G$ of length at most~$\ell$.
Now, we try all possibilities of selecting~$k$ cycles from this collection, yielding the running time $\binom{N}{k} \in n^{\OO(k^6)}$.

Now this algorithm can be turned into an algorithm that finds a collection of~$k$ pairwise vertex-disjoint cycles with total length at most~$\ell$ by self-reduction as follows.
We repeatedly attempt to remove edges and check if the solution still exists after an edge is removed.
These cause only a polynomial (in~$n$) overhead in the running time.
This concludes the proof for \msc.

For \msedc, we observe that if the above algorithm for \msc concludes that the graph contains~$k$ vertex-disjoint cycles of total length at most~$\ell$, then these~$k$ cycles are also edge disjoint. Otherwise, we have that
the number of cycles of length at most~$\ell$ is at most~$N$ and we can test for each selection of~$k$ of them whether they are edge-disjoint in~$N^{\OO(k)}$ time. This concludes the proof.
\end{proof}

Next, we generalize \Cref{thm:minsumxpalg} for \mvc{} and \mvedc{}.

\begin{theorem}
\label{the:minvecxpalg}
\mvc{} and \mvedc{} can be solved in~$n^{\OO(k^6)}$ time.
\end{theorem}
\begin{proof}
Let us give a decision algorithm.
This can be turned into an algorithm for finding a solution similarly as in the proof of \Cref{thm:minsumxpalg}.

Let~$(\ell_1,\ldots,\ell_k)$ be the given vector, and assume without loss of generality that $\ell_1 \le \ell_i$ for all $i$.
Similarly as in the proof of \Cref{the:minvecxpalg}, we apply the combination of \Cref{lem:mainlemminsumcyc} and \Cref{lem:pathgen} with the parameters~$k$ and~$\ell_1$ to, in time $n^{\OO(k^5)}$, either conclude that~$G$ contains a collection of $k$ pairwise (vertex- and edge-)disjoint cycles of total length at most $\ell_1$, or enumerate all (at most $N \in n^{\OO(k^5)}$) cycles of~$G$ of length at most $\ell_1$.
In the former case, we are done as any such collection gives us a solution, and in the second case, we branch on the cycles to guess the first cycle, delete its vertices (or edges in the case of \mvedc) 
and recursively apply the algorithm to find the $k-1$ other cycles. 
This branching algorithm runs in a total time of $(n^{\OO(k^5)})^k = n^{\OO(k^6)}$.
\end{proof}

\section{Lower bound for \msc}\label{sec:lb}
In this section, we prove our computational lower bound for \msc. For convenience, we restate our result. 

\thmminsumlb*

\begin{proof}
    We reduce from \mcc which is well-known to be \classW{1}-complete~\cite{CyganFKLMPPS15}.
    We recall that the task of the problem is, given a graph $G$ whose set of vertices is partioned into~$\ell$ sets $V_1,\ldots,V_\ell$ (called \emph{color classes}), to decide whether $G$ has a clique of size $\ell$ with exactly one vertex from each color class.
    
    Let~$(G,\ell)$ be an instance of \mcc{} and let~$V_i = \{v^i_1,v^i_2,\ldots,v^i_{\nu_i}\}$ be the set of vertices of color~$i$ for each~$i \in [\ell]$.
    Let~$\nu$ be the maximum number of vertices of any color, let~$\Delta$ be the maximal degree of~$G$, and let~$\gamma = (\nu-1)(3\Delta+1)-1$.
    For each color~$i \in [\ell]$, we create a vertex-selection gadget as follows.
    We start with~$6\nu$ vertices~$w_a^{i,1}, u^{i,1}_a, w_a^{i,2}, u^{i,2}_a, w_a^{i,3}$, and~$u^{i,3}_a$ for each~$a \in [\nu]$.
    We add an edge~$\{w_a^{i,j},u^{i,j}_a\}$ for each~$i \in [\ell]$, $j \in [3]$, and each~$a \in [\nu]$.
    Next, for each~$i \in [\ell]$, $j \in [3]$, $a \in [\nu]$, 
    we add~$3\Delta-1$~vertices~$v^{i,j}_{a,b}$ where ~$b \in [3\Delta-1]$.
    We add the edge~$\{v^{i,j}_{a,b},v^{i,j}_{a,b+1}\}$ for each~${i \in [\ell]}$, ${j \in [3]}, {a \in [\nu]}$, and each~${b \in [3\Delta-2]}$.
    Moreover, we add the edges~$\{u^{i,j}_{a},v^{i,j}_{a,1}\}$ and~$\{v^{i,j}_{a,3\Delta-1},w^{i,j}_{a+1}\}$ for each~${i \in [\ell]}$, ${j \in [3]},$ and~$a \in [\nu-1]$ and the edges~$\{u^{i,j}_{\nu},v^{i,j}_{\nu,1}\}$ and~$\{v^{i,j}_{\nu,3\Delta-1},w^{i,j \bmod 3 + 1}_{1}\}$ for each~$i \in [\ell]$ and each~$j \in [3]$. \
    Note that the entire constructed cycle has length~$3\nu(3\Delta+1)$.
    Finally, we add paths of length~$2\gamma$ between all pairs of~$w$ and~$u$ vertices that have distance exactly~$(\nu-1)(3\Delta+1)-1 = \gamma$.
    We call these paths \emph{chords}.
    \Cref{fig:vsg} gives an illustration of the above construction.
    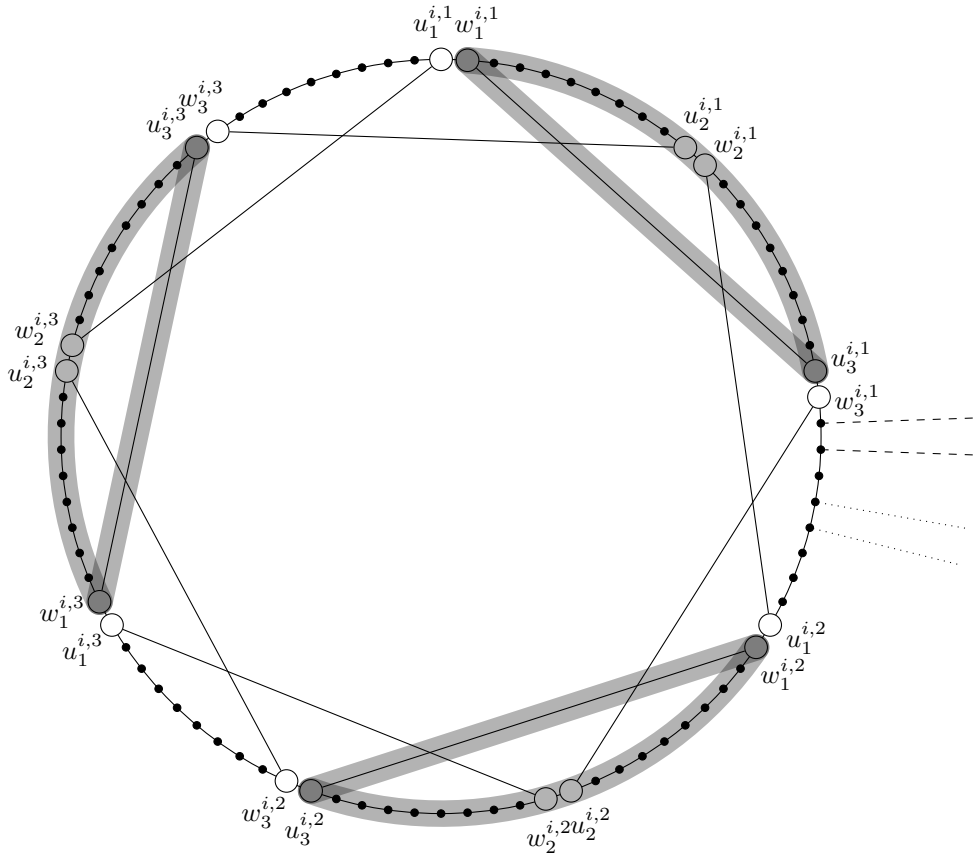
\begin{figure}[t]
        \centering
        \begin{tikzpicture}
            \draw circle (5);
            \foreach \x in {1,2,3} {
                \foreach \y in {1,2,3} {
                    \node[circle,inner sep=3pt,draw,fill=white] (w\x\y) at (246-120*\x-40*\y:5) {};
                    \node[circle,inner sep=3pt,draw,fill=white] (u\x\y) at (246-120*\x-40*\y+4:5) {};
                    \node at (246-120*\x-40*\y-1:5.5) {$w^{i,\x}_\y$};
                    \node at (246-120*\x-40*\y+5:5.5) {$u^{i,\x}_\y$};
                    \foreach \z in {1,2,3,4,5,6,7,8} {
                       \node[circle,draw,fill=black,inner sep=1pt] (v\x\y\z) at (246-120*\x-40*\y-4*\z:5) {};
                    }
                }
            }
            \foreach \x in {1,2} {
                \pgfmathtruncatemacro\a{\x + 1};
                \foreach \y in {1,2} {
                    \pgfmathtruncatemacro\z{\y + 1};
                    \draw (w\x\z) to (u\a\y);
                }
                \pgfmathtruncatemacro\z{1};
                \pgfmathtruncatemacro\y{3};
                \draw (w\x\z) to (u\x\y);
            }
            \pgfmathtruncatemacro\x{3};
            \pgfmathtruncatemacro\a{1};
            \foreach \y in {1,2} {
                \pgfmathtruncatemacro\z{\y + 1};
                \draw (w\x\z) to (u\a\y);
            }
            \pgfmathtruncatemacro\z{1};
            \pgfmathtruncatemacro\y{3};
            \draw (w\x\z) to (u\x\y);

            \draw[path] (w11.center) -- (u13.center);
            \draw[path] ($(86:5)$) arc (86:10:5);
            
            \draw[path] (w21.center) -- (u23.center);
            \draw[path] ($(206:5)$) arc (206:130:5);
            
            \draw[path] (w31.center) -- (u33.center);
            \draw[path] ($(326:5)$) arc (326:250:5);

            \draw[dashed] ($(2:7)$) -- (v131);
            \draw[dashed] ($(358:7)$) -- (v132);
            \draw[dotted] ($(350:7)$) -- (v134);
            \draw[dotted] ($(346:7)$) -- (v135);
        \end{tikzpicture}
        \caption{An example of the vertex-selection gadget with~$\nu = \Delta = 3$  for one color~$i$. The chords depict paths of length~$2(\nu-1)(3\Delta-1) = 32$ and one of the ways to pick three vertex-disjoint cycles using three chords is highlighted. This solution picks vertex~$v^i_3$. Paths encoding two edges incident to~$v^i_3$ are shown by dashed and dotted lines, respectively.}
        \label{fig:vsg}
    \end{figure}

    Next, we encode the edges of the input graph.
    For each color~$i\in [\ell]$ and each vertex~$v^i_a \in V_i$, arbitrarily assign a distinct number from~$[\Delta]$ to each incident edge.
    Let~$f(v^i_a,e)$ be the assigned number.
    For each edge~$e=\{v^i_a,v^j_b\} \in E$, we add paths of length~$\lceil\frac{8}{5}\gamma\rceil$ between~$v^{i,1}_{a,3f(v^i_a,e)-2}$ and~$v^{j,1}_{b,3f(v^j_b,e)-2}$ and between~$v^{i,1}_{a,3f(v^i_a,e)-1}$ and~$v^{j,1}_{b,3f(v^j_b,e)-1}$.
    We say that these two paths \emph{encode} the edge~$e$.
    Finally, we set~$k = 3\ell + {\ell \choose 2}$.

    Since the above construction takes polynomial time to compute,~$k \leq 4\ell^2$, and each vertex has degree at most three, it only remains to prove that the input instance contains a multicolored clique (of size~$\ell$) if and only if the constructed graph contains~$k$ vertex-disjoint cycles of total length at most~$L=9\ell\gamma + {\ell \choose 2}2(\lceil\frac{8}{5}\gamma\rceil+1)$. 
    For the backward direction, note that all cycles that use vertices from more than one vertex-selection gadget have length at least~$2(\lceil\frac{8}{5}\gamma\rceil+1) > 3 \gamma$.
    Moreover, all cycles within one vertex-selection gadget that use more than one chord have length at least~$4\gamma$.
    Hence, any solution of total length at most~$L$ contains at least~$3\ell$ cycles of total length at most~$9\ell\gamma$.
    By construction within one vertex-selection gadget, one can pick at most~$3$ vertex-disjoint cycles of average length at most~$3\gamma$ and this is only achievable if one picks three equally spaced chords and all vertices from the initial cycle of the vertex-selection gadget with the exception of three sets of~$v$ vertices between two consecutive~$w$ and~$u$ vertices as depicted in \Cref{fig:vsg}.
    We say that such a solution avoiding the~$v$ vertices between~$w^{i,1}_a$ and~$u^{i,1}_{a+1}$ (or~$u^{i,2}_1$ if~$a = \nu$) \emph{picks} vertex~$v^i_a$.
    If we select a solution that picks a vertex for each vertex-selection gadget, then it remains to find~${\ell \choose 2}$ vertex-disjoint cycles that use vertices from at least two vertex-selection gadgets of total length at most~${\ell \choose 2}2(\lceil\frac{8}{5}\gamma\rceil+1)$.
    Since each path between two vertex selection gadgets is of length~$\lceil\frac{8}{5}\gamma\rceil$, each vertex in a vertex-selection gadget has at most one incident edge leaving the vertex-selection gadget, and two paths between vertex-selection gadgets have adjacent endpoints if and only if they encode an edge in~$G$, it follows that any solution of total length at most~$L$ contains~${\ell \choose 2}$ pairs of paths encoding edges of~$G$.
    By construction, this corresponds to a set of~${\ell \choose 2}$ edges between~$\ell$ vertices of different colors, that is, to a multicolored clique (of size~$\ell$).

    For the forward direction, note that if there exists a multicolored clique, then choosing the respective vertex in each vertex-selection gadget and taking the cycles encoding all edges between the chosen vertices results in a set of~$k$ vertex-disjoint cycles of total length exactly~$L$.
    This concludes the proof.
\end{proof}

Since two cycles in subcubic graphs are edge-disjoint if and only if they are vertex-disjoint, we can defer the same hardness result for \msedc.

\begin{corollary}
    \msedc is \classW{1}-hard in subcubic graphs with unit edge weights when parameterized by~$k$.
\end{corollary}

\section{Packings  shortest cycles in planar graphs}\label{sec:planar}
In this section, we consider packings of disjoint shortest cycles in planar graphs and prove \Cref{thm:mincycles-vd,thm:mincycles-ed}. In \Cref{sec:laminar}, we construct a tree structure of a laminar family of shortest cycles that represent all minimum cycles. In~\Cref{sec:dsc-FPT}, we prove \Cref{thm:mincycles-vd} and \Cref{sec:dsce-kern} contains the proof of~\Cref{thm:mincycles-ed}. 

\subsection{Laminar decomposition for shortest cycles}\label{sec:laminar}
In this subsection, we construct a laminar family of disjoint cycles representing all shortest cycles in a planar graph and decompose this family into a tree. Throughout this subsection, we assume that considered graphs are not forests, that is, they have cycles.

We use the following folklore properties of shortest cycles which we prove for completeness. Given two distinct cycles $C_1$ and $C_2$ of a graph $G$ with a nonempty intersection, we say that~$C_1$ and~$C_2$ \emph{touch} if $C_1\cap C_2$  is a path (possibly trivial, that is, having a single vertex).

\begin{lemma}\label{lem:trivial}
Let $G$ be a weighted graph and let $C_1$ and $C_2$ be distinct shortest cycles with at least one common vertex. Then 
\begin{itemize}
\item either $C_1$ and $C_2$ touch,
\item or $V(C_1\cap C_2)=\{s,t\}$ for distinct $s$ and $t$ at distance $g(G)/2$ in both cycles, and 
$C_1=P_1\cup P_2$ and $C_2=Q_1\cup Q_2$ where $P_1$, $P_2$, $Q_1$, and $Q_2$ are distinct internally vertex-disjoint $s$-$t$-paths of length $g(G)/2$.
\end{itemize}
\end{lemma}

\begin{proof}
Since~$C_1$ and $C_2$ do not touch, $C_1$ has two distinct internally vertex-disjoint paths~$P_1$ and~$P_2$ such that the endpoints of these paths are in $C_2$ and the internal vertices and all the edges are outside $C_2$. Then, the length of one of these paths, say, $P_1$ is upper-bounded by~$g(G)/2$. Let $s$ and $t$ be the endpoints of $P_1$. Denote by $Q_1$ and $Q_2$ two distinct $s$-$t$-paths in~$C_2$. Since~$S_1=P_1\cup Q_1$ and $S_2=P_1\cup Q_2$ are cycles, the length of $S_1$ and $S_2$ is at least~$g(G)$. Therefore, the paths $P_1$, $Q_1$, and $Q_2$ have the same length $g(G)/2$. Also, we have that the length of $P_2$ is~$g(G)/2$. Thus, $P_1$, $P_2$, $Q_1$, and $Q_2$ are distinct internally vertex-disjoint $s$-$t$-paths of length~$g(G)/2$. This concludes the proof.
\end{proof}

We need the following additional notation for plane graphs.

\begin{definition}[\textbf{Laminar family}]
We say that two cycles $C_1$ and $C_2$ in $G$ \emph{cross} if $C_1$ has at least one edge 
in the internal face 
of $C_2$ and, symmetrically, $C_2$ has at least one edge in the internal face of $C_1$. A family $\mathcal{C}$ of cycles is \emph{laminar} if cycles in $\mathcal{C}$ do not cross.
\end{definition}

Let $G$ be a weighted plane graph. 
We introduce the following partial order on the family of all shortest cycles of $G$. For two shortest cycles $C_1$ and $C_2$, we write $C_1\leq C_2$ if every vertex and edge of~$C_1$ is a vertex or an edge of~$C_2$ or is embedded in the internal face of $C_2$. We also write $C_1<_v C_2$ if~$C_1\leq C_2$ and $V(C_1)\cap V(C_2)=\emptyset$ (that is, $C_1$ is completely inside $C_2$), and we write $C_1<_e C_2$ if~$C_1\leq C_2$ and~$E(C_1)\cap E(C_2)=\emptyset$ (i.e., $C_1$ and $C_2$ may share vertices but not edges). 

For a cycle $C$ in $G$, we denote by $G_C$ the subgraph of $G$ induced by the vertices of $C$ and the vertices of $G$ embedded in the internal face of $C$.  
We say that $G$ is \emph{clean} if each vertex and each edge of $G$ is included in some shortest cycle. 
We show that a clean graph always has a facial shortest cycle in the same way as other uncrossable families~\cite{schlomberg2023packing}.

\begin{lemma}\label{lem:facial}
Let $G$ be a clean weighted plane graph. Then, there is an internal face $f\in F(G)$ whose frontier is a shortest cycle.  
\end{lemma}

\begin{proof}
Let $C$ be a shortest cycle that is minimal with respect to the partial order $(\leq)$. We claim that the internal face of $C$ is a face of $G$. For the sake of contradiction, assume  that this is not the  case. Then, $G$ either has a vertex $v$ embedded in the internal face of $C$ or an edge $e$ with its endpoints in $C$ which is embedded in the internal face of $C$. In both cases, because $G$ is clean, there is a cycle $C'$ containing either $v$ or $e$. Since $C$ is minimal, $C'\not\leq C$. Therefore, $C$ and $C'$ cross and $C'$ has either a vertex or an edge in the external face of $C$. By~\Cref{lem:trivial}, $V(C\cap C')=\{s,t\}$ for distinct $s$ and $t$ at distance $g(G)/2$ in both cycles, and 
$C=P_1\cup P_2$ and $C'=Q_1\cup Q_2$ where $P_1$, $P_2$, $Q_1$, and $Q_2$ are distinct internally vertex-disjoint $s$-$t$-paths of length $g(G)/2$. Notice that either $Q_1$ or $Q_2$ contain $v$ or $e$. By symmetry, assume that this holds for~$Q_1$. Because $P_1$, $P_2$, and $Q_1$ are distinct internally vertex-disjoint, we have that~$Q_1$ is drawn in the internal face of $C$. This means that $S=P_1\cup Q_1$ is a shortest cycle and $S\leq C$. However, this contradicts the assumption that $C$ is minimal. This concludes the proof.   
\end{proof}

Let $C$ be a shortest cycle in a weighted plane graph $G$. We say that $C$ is \emph{splittable} if there are two vertices $s,t\in V(C)$ at distance $g(G)/2$ from each other in $C$ such that $G_C$ has an $s$-$t$-path $P$ of length~$g(G)/2$ distinct from the two $s$-$t$-paths in $C$; we say that $C$ is \emph{unsplittable}, otherwise. We call $s$ and $t$ \emph{poles}. 
Notice that for poles $s$ and $t$ on a splittable shortest cycle, the distance between them in $G$ is $g(G)/2$. We need the following observation about splittable cycles.

\begin{lemma}\label{lem:split}
Let $C$ be a splittable shortest cycle in a weighted plane graph $G$ with poles~$s$ and~$t$. Then, any two distinct shortest $s$-$t$-paths in $G_C$ are internally vertex-disjoint and $\{s,t\}$ is a unique pair of poles on $C$.
\end{lemma}

\begin{proof}
Suppose that $P_1$ and $P_2$ are distinct $s$-$t$-paths in $G_C$ of length $g(G)/2$. To show that~$P_1$ and~$P_2$ are internally vertex-disjoint, note that if $P_1$ and $P_2$ have a common vertex distinct from $s$ and $t$ then~$P_1\cup P_2$ has a cycle whose length is less than $g(G)$. This proves that any two shortest $s$-$t$-paths in $G_C$ are internally vertex-disjoint.

Consider the second property and let $P$ be a shortest $s$-$t$-path in $G_C$ that is distinct from the two $s$-$t$ paths along $C$. Assume that there is a pair of vertices $\{s',t'\}\neq \{s,t\}$ on $C$ at distance~$g(G)/2$ such that there is an $s'$-$t'$-path $P'$ of length $g(G)/2$ in $G_C$ that is not a path in $C$. Note that~$\{s',t'\} \cap \{s,t\} = \emptyset$. Otherwise, one of the two $s$-$t$-paths in $C$ or one of the two $s'$-$t'$-paths in $C$ is strictly shorter than~$g(G)/2$.  Let $Q$ be the $s$-$t$-path in $C$ containing $s'$ and let $Q'$ be the  $s'$-$t'$-path in $C$ containing $s$. 
By the first claim of the lemma, $P$ and $Q$ are internally vertex-disjoint, and, similarly, $P'$ and $Q'$ are internally vertex-disjoint as well. Then, $S=P\cup Q$ and $S'=P'\cup Q' $ are shortest cycles. Notice that $S\cap S'$ contains an~$s$-$s'$-path in $C$ which is nontrivial (contains at least two vertices) because~$s\neq s'$. Also, because $P$ and~$P'$ are paths in $G_C$, the paths have a common vertex $v$ in the internal face of $C$ and, therefore, $S\cap S'$ is not a path. By \Cref{lem:trivial}, we have that $S$ and $S'$ intersect in exactly two vertices, implying that~$s=s'$; a contradiction. This concludes the proof.   
\end{proof}

\begin{figure}[t]
\centering
\scalebox{0.7}{
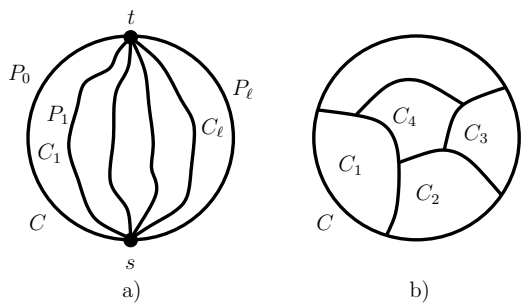}
\caption{Construction of  $\cP_s(C)$ and $\cC_s(C)$ for a splittable cycle (a) and construction of $\cC_u(C)=\{C_1,C_2,C_3,C_4\}$ for an unsplittable cycle (b). 
}
\label{fig:decomp}
\end{figure} 

Let $C$ be a splittable shortest cycle in a weighted plane graph $G$ with poles $s$ and $t$. We denote by~$\cP_s(C)=\{P_0,\ldots,P_\ell\}$ the inclusion-maximal family of distinct $s$-$t$-paths in $G_C$. By \Cref{lem:split}, $\cP_s(C)$ is unique. We assume that the paths on $\cP_s(C)$ are ordered in the clockwise order from the perspective of $s$ (see \Cref{fig:decomp}(a)) and $C=P_0\cup P_\ell$. For $i\in[\ell]$, we define the cycle $C_i=P_{i-1}\cup P_i$, and set~$\cC_s(C)=\{C_1,\ldots,C_\ell\}$. We use the following crucial property of splittable shortest cycles.

\begin{lemma}\label{lem:split-decomp}
Let $C$ be a splittable shortest cycle in a weighted plane graph $G$. Then, for any shortest cycle $S$ in $G_C$, either $S=P\cup Q$ for two paths~$P,Q\in\cP_s(C)$ or $S\leq R$ for some cycle~$R\in \cC_s(C)$.
\end{lemma}

\begin{proof} 
Let $S$ be a shortest cycle in $G_C$ and assume that $S\not\leq R$ for any $R\in \cC_s(C)$. 
Suppose that there is $R\in\cC_s(C)$ such that $S$ and $R$ cross. Then by \Cref{lem:trivial}, 
$V(S\cap R)=\{s',t'\}$ for distinct~$s'$ and $t'$ at distance $g(G)/2$ in both cycles, and 
$S=P_1\cup P_2$ and $R=Q_1\cup Q_2$ where~$P_1$,~$P_2$,~$Q_1$, and~$Q_2$ are distinct internally vertex-disjoint $s'$-$t'$-paths of length $g(G)/2$. Because $S$ and $R$ cross, we can assume without loss of generality that the edges and internal vertices of $P_1$ are embedded in the internal face of~$R$ and  the edges and internal vertices of~$P_2$ are in the external face of $R$. If $\{s',t'\}=\{s,t\}$ then we have that $P_1$ is an $s$-$t$-path in~$G_C$ internally vertex disjoint with the $s$-$t$-paths of $\cP_s(C)$ forming $R$. Then $P_1$ is internally vertex disjoint with every path of $\cP_s(C)$. However, this would contradict the maximality of~$\cP_s(C)$. Thus, $\{s',t'\}\neq\{s,t\}$. Moreover, note that $\{s',t'\}\cap \{s,t\}=\emptyset$ as otherwise, we would contradict the fact that~$Q_1$ and~$Q_2$ are paths of length $g(G)/2$. But now, we have that $s$ and  $t$ are embedded in distinct faces of~$S$. This is impossible because $s,t\in V(C)$ and $C$ is the frontier of the external face of~$G_C$ (recall that $S$ cannot contain $s$ or $t$ as it already contains $s'$ and $t'$). This contradiction shows that~$S$ and $R$ do not cross for any $R\in \cC_s(C)$. 

Since $S$ and $R$ do not cross for any $R\in \cC_s(C)$, the definition of~$\cC_s(C)$ implies that~${S=P\cup Q}$ for two paths $P,Q\in\cP_s(C)$. This concludes the proof.
\end{proof}

Now we deal with unsplittable shortest cycles. Consider an unsplittable shortest cycle $C$ in a weighted plane graph $G$. We define $\cC_u(C)$ to be the set of all shortest cycles of~$G_C$ distinct from~$C$, that are maximal cycles of this type with respect to $(\leq)$ (see~\Cref{fig:decomp}(b)). We use the following property of~$\cC_u(C)$.

\begin{lemma}\label{lem:unsplit-decomp}
Let $C$ be an unsplittable shortest cycle in a weighted plane graph $G$. Then $\cC_u(C)$ is laminar and for any shortest cycle $S\neq C$ in $G_C$, $S\leq R$ for some cycle $R\in \cC_u(C)$.
\end{lemma}
\begin{proof}
To show that $\cC_u(C)$ is laminar, consider distinct $R_1,R_2\in\cC_u(C)$ and assume that~$R_1$ and~$R_2$ cross. By \Cref{lem:trivial}, $V(R_1\cap R_2)=\{s',t'\}$ for distinct $s'$ and $t'$ at distance $g(G)/2$ in both cycles, and~$R_1=P_1\cup P_2$ and~$R_2=Q_1\cup Q_2$ where $P_1$, $P_2$, $Q_1$, and $Q_2$ are distinct internally vertex-disjoint~$s'$-$t'$-paths of length $g(G)/2$. Because~$R_1$ and $R_2$ cross, we can assume without loss of generality that the edges and internal vertices of $Q_1$ are embedded in the internal face of~$R_1$ and the edges and internal vertices of $P_2$ are in the external face of $R_2$. Consider~${R=P_1\cup Q_2}$. We have that $R_1\leq R$ and $R_2\leq R$. Recall that~$R_1$ and $R_2$ are maximal with respect to $(\leq)$ shortest cycles distinct from $C$. Since $R\notin \cC_u(C)$, $R=C$. However, we have that~$s,t\in V(C)$ and $G_C$ has four internally vertex-disjoint $s$-$t$-paths of length~$g(G)/2$. This means that $C$ is splittable and contradicts our assumptions. Thus, $\cC_u(C)$ is laminar.

To see the second claim, it is sufficient to observe that already by the definition of $\cC_u(C)$,
for any shortest $S\neq C$ in $G_C$, $S\leq R$ for some maximal (with respect to $(\leq)$) shortest cycle of~$G_C$ that is distinct from $C$. This concludes the proof.
\end{proof}

Now we are ready to construct a tree representing a laminar family of shortest cycles.

Let $G$ be a clean planar weighted graph. Consider an arbitrary embedding of $G$ in the plane. 
By \Cref{lem:facial}, there is a face $f$ in this embedding whose frontier is a shortest cycle $C$. Then, there is another planar embedding of $G$ such that $C$ is a facial cycle of the external face of the embedded graph. We fix this embedding and construct the rooted tree $\cT(G)$, called the \emph{Laminar Shortest Cycle Tree} (LSCT), whose nodes are shortest cycles. 

We start constructing $\cT(G)$ from the facial cycle $C$ of the external face which is set to be the root. Then, we construct $\cT(G)$ by the recursive analysis of shortest cycles $C$ that are current leaves of the already constructed tree.
\begin{itemize}
\item If $C$ is a facial cycle for an internal face then we set $C$ to be a leaf of $\cT(G)$.
\item If $C$ a splittable cycle then we construct $\cC_s(C)$ and set the cycles of $\cC_s(C)$ to be the children of $C$ in $\cT(G)$; we say that $C$ is an $S$-node in this case.
\item If $C$ is an unsplittable cycle then we construct $\cC_u(C)$ and set the cycles of $\cC_u(C)$ to be the children of $C$ in $\cT(G)$; we say that $C$ is a $U$-node.
\end{itemize}

The properties of
LSCTs are summarized in the following lemma.

\begin{lemma}\label{lem:tree}
Suppose that the LSCT $\cT(G)$ is constructed for a clean planar weighted graph $G$. Then the nodes of $\cT(G)$ form a laminar family of shortest cycles such that 
\begin{itemize}
\item[(i)] for any two nodes $C$ and $C'$ of $\cT(G)$, $C\leq C'$ in the planar embedding used in the construction of $\cT(G)$ if and only if $C$ is a descendant of $C'$ in $\cT(G)$,
\item[(ii)] for any shortest cycle $C$ of $G$, either $C$ is a node of $\cT(G)$ or there is an $S$-node $R$ of~$\cT(G)$ and two paths $P,Q\in \cP_s(R)$ such that $C=P\cup Q$,
\item[(iii)] the facial shortest cycles of $G$ distinct from the root of $\cT(G)$ are the leaves of $\cT(G)$.
\end{itemize}
Furthermore, the LSCT
can be constructed in polynomial time.
\end{lemma}

\begin{proof}
Properties (i)--(iii) follow directly from \Cref{lem:split-decomp}, \Cref{lem:unsplit-decomp}, and the construction of $\cT(G)$. To prove that $\cT(G)$ can be constructed in polynomial time, recall that a planar embedding $G$ can be found in linear time~\cite{HopcroftT74}. Then we can consider all the faces and find a face~$f$ whose frontier is a shortest cycle $C$. This allows us to construct an embedding of~$G$ where the frontier of the external face is a shortest cycle, in polynomial time. Then, following the steps of the algorithm, for each $C$ that is a current leaf of the already constructed tree, we verify whether $C$ is splittable or not. This can be done in polynomial time using Dijkstra's algorithm~\cite{Dijkstra59}. If~$C$~is splittable, then we construct the unique family of paths~$\cP_s(C)$ by greedily finding shortest $s$-$t$-paths. This allows us to construct $\cC_s(C)$ in polynomial time. If~$C$~is unsplittable, then we can, for example, list all shortest cycles in $G_C$ using the well-known fact that the number of such cycles is quadratic in the number of vertices and they can be enumerated in polynomial time (e.g., by using the algorithm of Karger and Stein~\cite{KargerS96} for the enumeration of all minimum cuts in the dual graph). Then, we can find all maximal shortest cycles with respect to ($\leq$) distinct from $C$ in $G_C$ and obtain $\cC_u(C)$ in polynomial time. Summarizing, we conclude that the overall running time is polynomial. This concludes the proof.   
\end{proof}

We conclude this section with some structural observations about solutions of \dsc and \dsce on planar graphs. Let $G$ be a clean weighted planar graph and let $k$ be a positive integer. Recall that we assume that~$G$ has a fixed planar embedding with the frontier of the external face being a shortest cycle. For LSCT $\cT(G)$, we define special cycle packings.
\begin{definition}
A packing $\cC=\{C_1,\ldots,C_k\}$ of $k$ (vertex/edge)-disjoint shortest cycles is \emph{$\cT$-maximal} if
\begin{itemize}
\item[(i)] the number of facial cycles of internal faces in the packing is maximum,
\item[(ii)] there is no packing $\cC'=\{C_1',\ldots,C_k'\}$ distinct from $\cC$ such that $C_i'\leq C_i$ for all $i\in[k]$.
\end{itemize}
\end{definition}

Our algorithm uses the following property. 

\begin{lemma}\label{lem:T-max}
Suppose that the LSCT $\cT(G)$ is constructed for a clean weighted planar graph~$G$ and assume that $G$ 
is embedded in the plane according to the construction of $\cT$.
Let~$\cC$ be a $\cT$-maximal packing of $k$ (vertex/edge)-disjoint shortest cycles. Then, $\cC$ is laminar and for each~$C\in\cC$, either $C$ is a facial cycle of an internal face or there is a~$C'\in \cC$ such that~$C'\leq C$.
\end{lemma}

\begin{proof}
If $\cC=\{C_1,\ldots,C_k\}$ is a packing of vertex-disjoint cycles, then $\cC$ is trivially laminar. Suppose that~$\cC$ is a packing of edge-disjoint cycles and assume that there are distinct $i,j\in[k]$ such that $C_i$ and~$C_j$ cross. Then by \Cref{lem:trivial}, $V(C_i\cap C_j)=\{s,t\}$ for distinct $s$ and $t$ at distance $g(G)/2$ in both cycles, and~$C_i=P_1\cup P_2$ and $C_j=Q_1\cup Q_2$ where $P_1$, $P_2$, $Q_1$, and~$Q_2$ are distinct internally vertex-disjoint~$s$-$t$-paths of length $g(G)/2$. Since the cycles cross, we can assume without loss of generality that the edges and internal vertices of $Q_1$ are in the internal face of $C_i$ and the edges and internal vertices of $P_2$ are in the internal face of~$C_j$.
Let~$C_i'=P_1\cup Q_1$ and $C_j'=P_2\cup Q_2$. Notice that~$C_i$ and~$C_j$ are not facial cycles and it holds that~$C_i'\leq C_i$ and~$C_j'\leq C_j$.
Consider~$\cC'$ obtained by the replacement of~$C_i$ and~$C_j$ with~$C_i'$ and~$C_j'$, respectively. We have that~$\cC'\neq \cC$ is a packing of $k$ edge-disjoint shortest cycles such that the number of facial cycles is at most as large as in~$\cC$. Moreover, because~$C_i'\leq C_i$ and~$C_j'\leq C_j$, the existence of~$\cC'$ contradicts that~$\cC$ is a $\cT$-maximal packing. Thus, $\cC$ is laminar. 

To prove the second claim, suppose that $C_i\in\cC$ is not a facial cycle of an internal face of $G$ for some~$i\in[k]$ and there is no $j\in[k]$ distinct from $i$ such that $C_j\leq C_i$. Consider $G_{C_i}$. Then by \Cref{lem:facial}, there is an internal face $f$ of $G_{C_i}$ whose frontier is a shortest cycle. Denote by~$C_i'$ such a cycle. We have that $C_i'\neq C_i$ and $C_i'\leq C_i$. Furthermore, because $\cC$ is laminar by the already proved first claim, $\cC'$ obtained from $\cC$ by the replacement of $C_i$ with $C_i'$ is a packing of disjoint cycles. However, the number of facial cycles of internal faces in $\cC'$ is bigger than the number of such cycles in $\cC$. This contradicts the assumption that $\cC$ is $\cT$-maximal. This proves the second claim and completes the proof of the lemma. 
\end{proof}

\subsection{FPT algorithm for \dsc}\label{sec:dsc-FPT}
In this subsection, we prove~\Cref{thm:mincycles-vd} which we restate here.

\thmmincyclesvd*

\begin{proof}
Let $(G,w,k)$ be an instance of \dsc where $G$ is a planar graph. If~$G$~is a forest, then we have a trivial no-instance.  Thus, we assume that this is not the case. 
Then, we preprocess~$G$ and delete every edge and every vertex that is not included in a shortest cycle as these edges and vertices are irrelevant for our problem. From now on, we assume that~$G$~is clean.

Next, we construct a planar embedding of $G$ such that the frontier of the external face is a shortest cycle. From this point, we assume that $G$ is a plane graph. We then construct the LSCT~$\cT(G)$.
Our aim is to find a $\cT$-maximal solution by using the properties given by \Cref{lem:T-max}. 

By \Cref{lem:tree} (ii), for any shortest cycle $C$ of $G$, either  
$C$ is a node of $\cT(G)$ or there is an $S$-node~$R$ of $\cT(G)$ and two paths $P,Q\in \cP_s(R)$ such that $R=P\cup Q$ where $C\neq P\cup Q$ and~$P\cup Q\notin \cC_s(R)$. If~$C$ is a shortest cycle of the second type, then we say that $C$ is a \emph{non-tree} cycle.

In the following, we wish to achieve the property that each cycle of a solution is a node of the tree. Since this may not always be the case, we ensure this property by modifying the tree at selected nodes. 

Let $C$ be an $S$-node of $\cT(G)$. We say that $C$ is \emph{large} if $|\cC_s(C)|\geq 3k+3$ and $C$ is \emph{small}, otherwise. We prove that for large $S$-nodes, we already have the desired property. 

\begin{claim}\label{cl:large}
Let $\cC$ be a $\cT$-maximal solution to $(G,w,k)$. Then for any large $S$-node~$C$ of~$\cT(G)$, if $\cC$ contains a cycle $R=P\cup Q$ for two distinct paths $P,Q\in\cP_s(C)$ 
then either $R=C$ or~$R\in\cC_s(C)$.
\end{claim}

\begin{claimproof}
Let $C$ be a large $S$-node of $\cT(G)$ with poles $s$ and $t$.
Let ${\cP_s(C)=\{P_0,\ldots,P_\ell\}}$ and~$\cC_s(C)=\{C_1,\ldots,C_\ell\}$. We assume that the paths in $\cP_s(C)$ are ordered in the clockwise order from the perspective of pole $s$ and $C_i=P_{i-1}\cup P_i$ for each $i\in [\ell]$, as shown in \Cref{fig:decomp}(a). Suppose that there is~$R\in \cC$ distinct from $C$ such that $R=P_i\cup P_j$ for $i,j\in[0,\ell]$ with $i<j-1$. Notice that $R$ is not a facial cycle. 

Consider $\cS=\cC\setminus \{R\}$. By the fact that $s,t\in V(R)$, \Cref{lem:tree}, and the laminarity of $\cC$,  we have that for every $S\in \cS$, $s,t\notin V(S)$ and either $S$ has no edges and vertices in the internal face of $C$ or there is $i\in[\ell]$ such that $S\leq C_i$. Because $C$ is large, $\ell\geq 3k+3$. Then by the pigeon hole principle, there is $i\in[2,\ell-1]$ such that there is no $S\in \cS$ with  $S\leq C_{i-1},C_i,C_{i+1}$. This means that $\cS'=\cS\cup\{C_i\}$ is a packing of $k$ vertex-disjoint shortest cycles. By \Cref{lem:facial}, there is a facial shortest cycle $R'$ of an internal face of $G_{C_i}$. As there is no $S\in\cS$ with $S\leq C_i$, we have that 
$\cC'=\cS\cup\{R'\}$ is a packing of $k$ vertex-disjoint shortest cycles. However, $\cC'$ is obtained from $\cC$ by replacing a non-facial cycle $R$ with a facial cycle $R'$. This contradicts the assumption that $\cC$ is a $\cT$-maximal solution and proves the claim.
\end{claimproof}

By \Cref{cl:large}, it remains to deal with non-tree cycles formed by two paths of $\cP(C)$ for small $S$-nodes. 
For this, we apply the \emph{random separation} technique~\cite{CaiCC06,CyganFKLMPPS15}. For simplicity, we give a randomized Monte Carlo procedure and explain how to derandomize this step in the conclusion of the proof of the theorem.

Let $\cP=\bigcup_C \cP_s(C)$ where the union is taken over all small $S$-nodes $C$ of $\cT(G)$, that is, $\cP$ is the family of all paths that may be parts of non-tree cycles in a solution. We color the paths of $\cP$ randomly by two colors \emph{red} and \emph{blue} in such  a way that a path $P\in \cP$ is colored red with probability $\frac{2}{3k+3}$ and~$P$ is colored blue with probability $\frac{3k+1}{3k+3}$. 
Consider a solution $\cC$ such that  
if~$\cC$ contains a non-tree cycle~$R=P\cup Q$ for two distinct paths $P,Q\in\cP_s(C)$ for 
a small $S$-node~$C$ of $\cT(G)$ then (i) $P$ and $Q$ are colored red and (ii) all other paths in $\cP_s(C)$ are blue. We say that such a solution $\cC$ is \emph{colorful}. 

To see the reason behind the coloring, assume that $\cC$ is a $\cT$-maximal solution to the considered instance. Let $R_1,\ldots,R_\ell$ be the non-tree cycles in $\cC$ formed by pairs of paths from~$\cP_s(C_1),\ldots,\cP_s(C_\ell)$ for small $S$-nodes~$C_1,\ldots,C_\ell$ of $\cT(G)$. Let $\cP'=\bigcup_{i=1}^\ell \cP_s(C_i)$. Clearly, $\ell\leq k$. The probability that all~$2\ell$ paths from $\cP'$ forming the non-tree cycles~$R_1,\ldots,R_\ell$ are red is at 
least~$\big(\frac{2}{3k+3}\big)^{2\ell}\geq \big(\frac{2}{3k+3}\big)^{2k}$.
Since $C_i$ is a small $S$-node for~$\cT(G)$, $|\cP_s(C_i)|\leq 3k+3$ for each~$i\in[\ell]$. Hence, the total number of paths in $\cP'$ that are not parts of non-tree cycles is at most~$(3k+1)\ell\leq (3k+1)k$. The probability that all these paths are colored blue is at least
\begin{equation*}
\Big(\frac{3k+1}{3k+3}\Big)^{(3k+1)k}\geq \Big(1-\frac{2}{3k+3}\Big)^{\frac{3k+3}{2}2k}\geq\Big(\frac{8}{27}\Big)^{2k}.
\end{equation*}
Thus, the probability that all paths in $\cP'$ forming non-tree cycles are red and all other paths are blue is at least $\Big(\frac{16}{81(k+1)}\Big)^{2k}$. 
This implies that if we try 
$N=\Big(\frac{81(k+1)}{16}\Big)^{2k}=k^{\OO(k)}$ random colorings, then the probability that the desired property that all paths in $\cP'$ forming non-tree cycles are red and all other paths are blue is not fulfilled for any of the colorings is at most~$\Big(1-\frac{1}{N}\Big)^N\leq\frac{1}{e}$. 

From now on, we assume that a coloring of the paths of $\cP$ is given such that if there is a solution, then a $\cal T$-maximal solution is colorful.
\begin{figure}[t]
\centering
\includegraphics[]{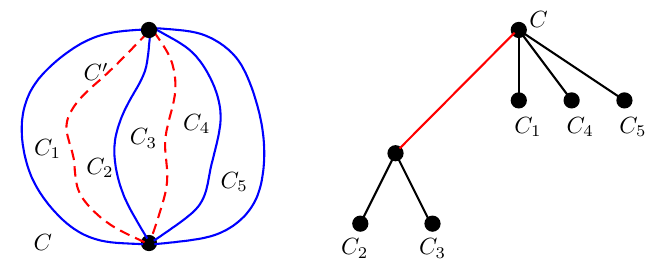}
\caption{Construction of  $\cT^*$. In the tree on the right side of the figure, the parent node of $C_2$ and $C_3$ corresponds to the cycle $C'$ (red dashed lines) on the left.}
\label{fig:T-star}
\end{figure}

We construct the tree $\cT^*$ from $\cT(G)$ by splitting selected small $S$-nodes. For every small $S$-node
$C$ such that there are two red (dashed) paths in~$\cP_s(C)$ forming a non-tree cycle $C'$ and all other paths in~$\cP_s(C)$ are blue (solid), we do the following (see~\Cref{fig:T-star}):
\begin{itemize}
\item create a new node $C'$ and make it a child of $C$,
\item make every cycle $S\in \cC_s(C)$ such that $S\leq C'$ a child of $C'$; the other cycles in $\cC_s(C)$ remain to be children of $C$.
\end{itemize}

Observe that every node of $\cT(G)$ is a node of $\cT^*$ and, by~\Cref{lem:tree} and the construction of $\cT^*$, we have that the  nodes of $\cT^*$ form a laminar family of shortest cycles such that the following two properties hold.

\noindent
{\bf Property (i):} For any two nodes $C$ and $C'$ of $\cT^*$, $C\leq C'$ in the planar embedding of $G$ if and only if~$C$ is a descendant of $C'$ in $\cT^*$. 

\noindent
{\bf Property (ii):} The facial shortest cycles of $G$ distinct from the root of $\cT^*$ are the leaves of~$\cT^*$.

Furthermore, if a colorful $\cT$-maximal solution to $(G,w,k)$ exists, then each cycle of the solution is a node of $\cT^*$. We use these properties to construct a recursive branching algorithm that finds a solution if a colorful $\cT$-maximal solution exists.

The algorithm takes $G$, $k$, and $\cT^*$ as its parameters; these parameters are modified in the recursive calls. 

In the first step, we construct the map graph $M$ with $G$ as its planar embedding whose  vertices are the faces of $G$ with their frontiers being shortest facial cycles. Then we call the algorithm from \Cref{prop:ind-map} to check whether $M$ has an independent set of size $k$. If such a set exists then the frontiers of the faces of $G$ in the independent set compose a packing of $k$ vertex-disjoint shortest cycles. Thus, we can conclude that $(G,w,k)$ is a yes-instance and stop. Assume that this is not the case.

Because there is no packing of $k$ vertex-disjoint shortest cycles that contains only facial cycles, by \Cref{lem:T-max}, we conclude that if there is a colorful $\cT$-maximal solution then any solution contains a non-facial cycle $C$ and a facial cycle~$C'$ such that $C'<_v C$. Moreover, by properties (i) and (ii) of $\cT^*$, there is a node $C$ of $\cT^*$ in the solution such that (a) $C$ is a non-facial cycle, (b) there is a facial cycle~$C'$ in the solution such that $C'<_v C$, and moreover, we can choose $C$ such that 
(c) any cycle $C''\neq C$ in the solution such that $C''\leq C$ is facial.
Our goal is to identify $C$ among all those shortest cycles that satisfy conditions (a)-(c), 

For this, we first mark the nodes of $\cT^*$ that are candidates to be $C$. Namely, we mark every node $C$ of $\cT^*$ such that there is a facial shortest cycle, that is, a leaf $C'$ of $\cT^*$, with the property that $C'<_v C$.   
\renewcommand{\star}{\alpha}
If no node was marked, then we conclude that there is no colorful $\cT$-maximal solution and so, we return that $(G,w,k)$ is a no-instance and stop. We now assume that this is not the case and consider the subgraph $\cT^{\star}$ of $\cT^*$ induced by the marked nodes. Notice that~$\cT^{\star}$ is a subtree of~$\cT^*$ with the same root as $\cT^*$, since any ancestor of a marked node is marked by property (i) of $\cT^*$. 

Denote by $\cL$ the set of leaves of $\cT^{\star}$ and let $\cU$ be the set of nodes of $\cT^\star$ that includes the root, the nodes of $\cL$, and all the nodes with at least two children.
Notice that because of \Cref{lem:facial}, for each~$C\in\cL$, there is a facial shortest cycle $C'<_v C$. Then by the laminarity of the cycles of~$\cT^\star$ and property (i) of~$\cT^\star$, we have that $|\cL|<k$. Otherwise, the map graph~$M$ would have an independent set of size $k$. However, this was already checked and ruled out. Since $|\cL|<k$, we have that 
$|\cU|\leq 2|\cL|-1\leq 2k-3$. Denote by $\cS$ the set of all nontrivial~$S_1$-$S_2$-paths
in $\cT^\star$ with their endpoints $S_1$ and $S_2$ in $\cU$ such that the internal vertices have degree two in~$\cT^\star$. Note that $|\cS|\leq |\cU|-1\leq 2k-4$. 

For each node $C$ in $\cT^\star$, consider the map graph $M_C$ such that $G_C$ is its planar embedding
and the vertices of $M_C$ are those faces of $G_C$ 
with shortest cycles as their frontiers and 
such that for any of their facial cycles $R$, we have that $R<_v C$, that is, all these faces are in the internal face of $C$ and the facial cycles do not touch $C$.

Recall that we aim to find the non-facial cycle $C$ satisfying conditions (a)--(c) stated above. For this, we guess the path in $\cS$ containing $C$ and the number $r$ of facial cycles $C'$ in the solution 
such that~$C'<_v C$. This is done by branching on all possible choices of a path in $\cS$ and a positive integer~$r\leq k-1$. 

Assume that $r$ and an $S_1$-$S_2$-path $P\in\cS$ are given. Since $S_1,S_2\in \cU$, we can assume without loss of generality that $S_2$ is a descendant of $S_1$ in $\cT^\star$. For each node $S$ in $P$, we check whether there is a packing of $r$ vertex-disjoint facial cycles $C'$ such that $C'<_v S$.
This is done by calling the algorithm from~\Cref{prop:ind-map} for $M_S$. If such a node does not exist, then we conclude that
either $C$ does not lie on $P$ or the solution does not have $r$ vertex-disjoint facial cycles $C'$ such that $C'<_v C$. So, we discard this choice of $P$ and $r$. 
On the other hand, if we found a cycle $S$ and $r$ vertex-disjoint facial cycles~$C'$ such that~$C'<_v S$ for~$r=k-1$, then these~$r$ facial shortest cycles together with~$S$ compose a packing of~$k$~vertex-disjoint shortest cycles. Thus, we return yes and stop.
Otherwise, among all nodes $S$ with the above property (lying on $P$ and having $r$ vertex-disjoint facial cycles $C'$ such that $C'<_v S$), we select the node $S^*$ closest to $S_2$ in $P$.

To argue that the choice of $S^*$ is feasible in the sense that 
if the guess of the path containing~$C$ and~$r$ was correct, then we can take~$C=S^*$,
assume that $C$ is a non-facial cycle in a hypothetical colorful~$\cT$-maximal solution $\cC$ 
such that $C$ is a node of $P$, conditions (a)--(c) are fulfilled, and there are exactly~$r$~facial cycles $C'$ in the solution with~$C'<_v C$. We additionally assume that~$\cC$~is a solution where $C$ is at a minimum distance from $S_2$ in $\cT^*$. Let $\cC'\subseteq \cC$ be the set of all the $r+1$ cycles $C'$ in the solution with~$C'\leq C$. By the choice of $S^*$, $S^*\leq C$ and~$G_{S^*}$ has a packing~$\cC''$ of $r+1$ vertex-disjoint shortest cycles that includes $S^*$ and $r$ facial cycles~$C'<_v S^*$. Consider $\cC^*=(\cC\setminus \cC')\cup \cC''$, that is, we replace the cycles of $\cC'$ with the cycles of $\cC''$. Because~$\cC$~is laminar and $S^*\leq C$, we  obtain that $\cC^*$ is a packing of $k$ vertex-disjoint cycles. Because we choose $\cC$ where $C$ is at a minimum distance from $S_2$ in $\cT^\star$, we obtain that $S^*=C$. This means that the choice of $S^*$ is feasible if $P$ and $r$ were correctly guessed.

Next,  we call our algorithm recursively. Since we would like to avoid rebuilding~$\cT^*$ because we are looking for a colorful solution defined for~$\cT^*$,
we construct 
new parameters $G'$, $k'$, and~$(\cT^*)'$ as follows:  
\begin{itemize}
\item construct $G'$ from $G$ by deleting the vertices and edges of $G$ that are embedded in the internal face of $C^*$,
\item set $k'=k-r$,
\item construct $(\cT^*)'$ from $\cT^*$ by deleting the nodes  that are proper descendants of $C^*$ (i.e., descendants of $C^*$ distinct from $C^*$) in $\cT^*$.
\end{itemize}
Observe that because we do not delete $C^*$, 
we can use $(\cT^*)'$ to represent the shortest cycles in the obtained graph $G'$. 

Then, we call the algorithm for $G'$, $k'$, and $(\cT^*)'$. If the algorithm returns yes, then we conclude that $(G,w,k)$ is a yes-instance and stop. If the algorithm returns no, then we discard the current choice of $P$ and $r$. 
To complete the description of the branching algorithm, we observe that we branch on at most $2k-4$ paths $P$ in $\cS$ and at most $k-1$ choices of $r$. If the algorithm fails to find a solution for all choices, then we conclude that there is no solution and return no. 

Since $r\geq 1$, we have that the depth of the recursion is at most $k$, that is, the algorithm is finite. Notice that if the algorithm concludes that there is a packing $\cC$ of $k$ vertex-disjoint shortest cycles, then it may happen that $\cC$ is not a $\cT$-maximal solution. However, it is sufficient for us that $\cC$ is a solution to~$(G,w,k)$. From the other side, if $(G,w,k)$ is a yes-instance, then the instance has a $\cT$-maximal solution. Hence, if there is a colorful solution of this type, then the algorithm returns yes. This concludes the correctness proof for the branching algorithm. 

Recall that we call this  algorithm for at most $N=\Big(\frac{81(k+1)}{16}\Big)^{2k}$ random colorings of $\cP$. If for one of the colorings, we obtain that $(G,w,k)$ is a yes-instance then we return yes and stop. Otherwise, if we fail to find a solution for all colorings, then we return no. The algorithm may return a false negative answer but the probability of this event is at most~$\frac{1}{e}$. 
This concludes the description and the correctness proof of the algorithm.

To evaluate the running time, note that we can verify whether $G$ is a forest, compute the girth, and 
delete the vertices and edges that are not included in shortest cycles in polynomial time by Dijkstra's algorithm~\cite{Dijkstra59}. Then, $\cT(G)$ and the corresponding planar embedding of~$G$ can be constructed in polynomial time by~\Cref{lem:tree}. Given $\cT(G)$, it is straightforward to construct the set of paths $\cP$ in polynomial time. Then, for each random coloring of $\cP$, it is also easy to construct~$\cT^*$  in polynomial time.  Then, we call our branching algorithm. 

For the evaluation of the running time of the branching algorithm, we note the map graph~$M$ can be constructed in polynomial time for the given planar embedding of $G$. The algorithm from~\Cref{prop:ind-map} runs in $2^{\OO(k)}\cdot n^2$ time. If the algorithm reports that $M$ has no independent set of size $k$, we construct~$\cT^\star$ from~$\cT^*$ and this can be done in polynomial time. Further, 
we branch on at most~$2k-4$ choices of~$P\in\cS$ and at most $k-1$ choices of $r$, that is, we have~$\OO(k^2)$~possibilities. Because $P$ has at most $n$ nodes, we call the algorithm from  \Cref{prop:ind-map} 
at most $n$ times. Thus, we can find the cycle $C^*$ (if exists) in  $2^{\OO(k)}\cdot n^3$ time.  
Then we call our branching algorithm recursively for $G'$, $k'$, and $(\cT^*)'$ which can be easily constructed in polynomial time. Since the depth of the recursion tree is at most $k$, we obtain that the total running time of the branching algorithm is $k^{\OO(k)}\cdot n^{\OO(1)}$.

We run the branching algorithm for at most $N=\Big(\frac{81(k+1)}{16}\Big)^{2k}=k^{\OO(k)}$ random colorings. Then the total running time is $k^{\OO(k)}\cdot n^{\OO(1)}$. Taking into account that the running time of the previous steps is polynomial, the overall running time is $k^{\OO(k)}\cdot n^{\OO(1)}$.

The described algorithm is a randomized Monte Carlo algorithm with one-sided error due to applying the random separation technique. More precisely, we consider random colorings of the paths in the set~$\cP=\bigcup_C \cP_s(C)$ where the union is taken over all small $S$-nodes $C$ of $\cT(G)$. Notice that we use a non-uniform probability distribution. To derandomize the algorithm, one can use the results 
of Fomin et al.~\cite{FominLPS16} tailored for dealing with random colorings of this type. More precisely, we can replace random colorings by \emph{separating collections}; we refer to~\citet{FominLPS16} for the definition of separating collections and the details of their construction. 
The obtained deterministic algorithm has a slightly worse running time. However, the overall running time still can be written as $k^{\OO(k)}\cdot n^{\OO(1)}$.
This concludes the proof.
\end{proof}

\subsection{Kernelization for \dsce}\label{sec:dsce-kern}
In this subsection, we prove~\Cref{thm:mincycles-ed} which we restate here.

\thmmincyclesed*

\begin{proof}

Let $(G, w, k)$ be an instance of \dsce where $G$ is a planar graph. We assume that $G$ is clean, as we can safely remove all vertices and edges that are not part of any shortest cycle.
If the graph is empty, we return a trivial no-instance. Otherwise, $G$ is clean and non-empty, therefore by \Cref{lem:facial}, there is also an internal face $f$ whose frontier is a shortest cycle, in some fixed planar embedding of $G$. We consider another embedding of $G$ where $f$ is the outer face, and compute the LSCT $\cT(G)$ as described in the proof of \Cref{thm:mincycles-vd}.

We first bound the number of leaves in $\cT(G)$.
\begin{claim}
    \label{claim:T_ed_leaves}
    If $\cT(G)$ has at least $4k$ leaves, then $(G, w, k)$ is a yes-instance.
\end{claim}
\begin{claimproof}
    Consider the graph $G_F$, where the vertex set is the set of shortest cycles in $G$ that are found in the leaves of $\cT(G)$, and a pair of vertices is adjacent if and only if the respective cycles share edges. Since the shortest cycles in the leaves of $\cT(G)$ are frontiers of internal faces in the embedding of~$G$, $G_F$ is a subgraph of the dual graph $G^*$ of $G$, and is therefore planar. By~\Cref{obs:ind-planar}, a planar graph with at least $4k$ vertices has an independent set of size $k$. The vertices of this independent set in~$G_F$ correspond to $k$ edge-disjoint shortest cycles in $G$, which proves the claim.
\end{claimproof}

By \Cref{claim:T_ed_leaves}, if there are at least $4k$ leaves in $\cT(G)$, we stop and return a trivial yes-instance. Thus, in the following we assume that $\cT(G)$ has less than $4k$ leaves.

For a shortest cycle $C$ in $G$, we say that its
\emph{extension} is any shortest cycle $C'$ with $C <_e C'$, and we say that $C'$ is its
\emph{lowest extension} if $C'$ is an extension of $C$ and for every other extension~$C''$ of~$C$, it holds that~$C' \le C''$.
\begin{claim}\label{claim:extensions}
Every shortest cycle $C$ in $G$ that has at least one extension, has a unique lowest extension.
\end{claim}

\begin{claimproof}
    Observe that for any shortest cycles $C' \le C''$ in $G$, if $C <_e C'$, then~$C <_e C''$.
    Consider first extensions of $C$ that are cycles in $\cT(G)$.  By \Cref{lem:tree}, all cycles $C'$ of $\cT(G)$ with~$C \le C'$ form a rooted subpath in~$\cT(G)$; by the above, the set of extensions of $C$ in~$\cT(G)$ forms a rooted subpath of this path. If the subpath is empty, then the root cycle itself shares edges with~$C$, in which case $C$ has no extension. Otherwise, let $C^*$ be the lowest cycle on this subpath; for each other extension~$C'$ in~$\cT(G)$, $C^* \le C'$. If $C^*$ is a U-node, by \Cref{lem:unsplit-decomp}~$C \le C'$ for some $C' \in \cC_u(C^*)$, and for every other shortest cycle $C''$ in $G_{C^*}$ with~$C \le C''$, it holds that~$C'' \le C'$.
    Since~$C' \in \cT(G)$ is not an extension of~$C$, there are no other extensions in $G_{C^*}$ and~$C^*$ is indeed the unique lowest extension.

    Now, consider the case where~$C^*$ is an S-node. Let $s,t$ be the poles of $C^*$ and let~${\cP_s(C) = \{P_0, \ldots, P_\ell\}}$ be the inclusion-wise maximal family of $s$-$t$-paths in $G_{C'}$ as in the construction of $\cC_s(C^*)$. By \Cref{lem:split-decomp}, either there exists $C' \in \cC_s(C^*)$ such that $C \le C'$, or~$C = P_i \cup P_j$ for some~$i < j \in [\ell]_0$.
    In the latter case, observe that $i > 0$ and $j < \ell$ since~$C$ shares no edges with $C^*$, and all non-tree cycles in $G_{C^*}$ that are extensions of $C$ are of the form~$P_{i'} \cup P_{j'}$ for some $i' < i$ and $j < j' \le \ell$. Hence, in this case $P_{i - 1} \cup P_{j + 1}$ is the lowest extension of $P_i \cup P_j$.
    Otherwise, if $C \le C'$ for some $C' \in \cC_s(C^*)$, then $C'$ shares edges with~$C$ by the choice of $C^*$ as~$C' \in \cT(G)$. Let $h \in [\ell]$ be such that $C' = P_{h - 1} \cup P_h$, let $i$ be the smallest index in~$\{h - 1, h\}$ such that $P_i$ shares edges with $C$, and let $j$ be the largest. By the same argument as in the previous case, $P_{i - 1} \cup P_{j + 1}$ is the lowest extension of $C$.
\end{claimproof}

Due to \Cref{claim:extensions}, for every shortest cycle $C$ in $G$ with at least one extension, we denote by~$L(C)$, the unique lowest extension of $C$. 
We say that the \emph{extension chain} $\cU(C)$ of $C$ is the sequence~$C_0, C_1, \ldots, C_r$ of shortest cycles in $G$ such that $C_0 = C$, $C_i = L(C_{i - 1})$ for each $i \in [r]$, and the cycle $C_r$ has no extension. 
Notice that, given $\cT(G)$, $\cU(C)$ can be found in polynomial time by tracing the path from $C$ to the root in the tree.

We now construct the set of \emph{marked} cycles $\cC_M$, which is a subset of all shortest cycles of $G$. We start by describing the set of \emph{base} cycles $\cC_B \subseteq \cC_M$, which contains
\begin{itemize}
    \item every leaf cycle of $\cT(G)$,
    \item every $U$-node that has at least two children,
    \item every $S$-node together with the cycles $P_0\cup P_{\ell - 1}$ and $P_1 \cup P_\ell$, where $\{P_0, \ldots, P_\ell\} = \cP_s(C)$.
\end{itemize}
If there exists a cycle $C \in \cC_B$ with a large enough extension chain, i.e., with $|\cU(C)| \ge k$, we stop and return a trivial yes-instance.
This is safe, since by definition all of the cycles in $\cU(C)$ are edge-disjoint.
Otherwise, we define $\cC_M$ to be all cycles of $\cC_B$ together with their extension chains, plus the root of~$\cT(G)$ and the set of all non-tree cycles: $\cC_M = \{r\} \cup \cC_N \cup \bigcup_{C \in \cC_B} \cU(C)$, where $\cC_N$ denotes the set of all shortest cycles in $G$ that are not in $\cT(G)$. We first observe that~$\cC_M$ is closed under taking the lowest extension.

\begin{claim}\label{claim:cM_closed}
    For every cycle $C \in \cC_M$, its lowest extension $L(C)$ also belongs to $\cC_M$, if it exists.
\end{claim}
\begin{claimproof}
    For each $C \in \cC_B$, the statement is immediate for $C$ and all cycles in $\cU(C)$, since their lowest extensions are added to $\cU(C)$ exhaustively by the definition.
    
    It remains to consider non-tree cycles. 
    Let $C$ be an S-node in $\cT(G)$ with ${\cP_s(C) = \{P_0, \ldots, P_\ell\}}$ and consider a non-tree cycle $C' = P_i \cup P_j$ for some $i < j \in [\ell]_0$.
    By the proof of \Cref{claim:extensions}, if $i > 0$ and~$j < \ell$, the lowest extension of $C'$ is $P_{i - 1} \cup P_{j + 1}$, which is either another non-tree cycle or the S-node~$C$; both are part of $\cC_M$ by definition. Finally, if $i = 0$ we claim that~$L(C') = L(P_0 \cup P_{\ell - 1})$, where~$P_0 \cup P_{\ell - 1} \in \cC_B$ and $L(P_0 \cup P_{\ell - 1}) \in \cC_M$. This holds since both lowest extensions are above $C$, while the intersection between $C$ and $C'$ is $P_0$, which is the same as the intersection between $C$ and $P_0 \cup P_{\ell - 1}$.
    The argument for $j = \ell - 1$ and $L(C') = L(P_1 \cup P_\ell)$ is analogous.
\end{claimproof}

By construction, we have an upper-bound on the size of $\cC_M$, which we prove in the next claim.

\begin{claim}\label{claim:marked_size}
    The size of $\cC_M$ is $\OO(k^2)$.
\end{claim}
\begin{claimproof}
    We first bound the number of base cycles, i.e., the size of $\cC_B$. 
    Recall that we have already bounded the number of leaves in $\cT(G)$ by 
    $4k-1$. This also implies that the number of internal vertices of $\cT(G)$ that have at least two children is less than $4k$. Therefore, the total number of $U$-nodes and $S$-nodes in $\cC_B$ is less than $4k$. By definition, $\cC_B$ contains two additional cycles for each $S$-node; from the above, the number of such cycles is less than $8k$. The total size of $\cC_B$ is therefore less than $16k$.
    
    For every $C \in \cC_B$, the size of its extension chain $\cU(C)$ is at most $k - 1$. Thus, the cycles of~$\cC_B$ together with their extension chains amount to less than $16k^2$ cycles. In addition to these cycles, $\cC_M$ contains the root and the non-tree cycles $\cC_N$. By the above, we have less than $4k$ $S$-nodes, and in total they have less than $8k$ children. Since an $S$-node with $p$ children contains at most $p \cdot (p + 1) / 2$ non-tree cycles, there are at most $4k \cdot (8k + 1)$ non-tree cycles in $G$ in total. The size of $\cC_M$ can be therefore upper-bounded by $52 \cdot k^2 = \OO(k^2)$.
\end{claimproof}

We now show the main property of the set of marked cycles $\cC_M$, that no other cycles are needed in order to find a $\cT$-maximal solution.

\begin{claim} \label{claim:marked}
    For every $\cT$-maximal solution $\cC$, it holds that $\cC \subseteq \cC_M$.
\end{claim}
\begin{claimproof}
    Consider a $\cT$-maximal solution $\cC$, and assume by contradiction that it is not contained in $\cC_M$. Let $C$ be the minimal cycle in $\cC$ with respect to ($\le$) that is not in $\cC_M$. Since $\cC_M$ contains all leaves and internal vertices of $\cT(G)$ with at least two children and all non-tree shortest cycles of $G$, we have that $C$ is a U-node in $\cT(G)$ with exactly one child $C'$. If $C' \in \cC$, then $C' \in \cC_M$ by the choice of $C$; but then $C' <_e C$, and $C$ is the lowest extension of $C'$, which contradicts the assumption that $C \notin \cC_M$, since by \Cref{claim:cM_closed}, $\cC_M$ is closed under taking the lowest extension.

    Therefore, $C' \notin \cC$, and let $\cC' = \cC \cup \{C'\} \setminus \{C\}$ be a set of $k$ shortest cycles obtained from~$\cC$ by replacing $C$ with $C'$. Since $\cC$ is $\cT$-maximal, it cannot be that $\cC'$ is a solution, thus $C'$ must share edges with some other cycle $C'' \in \cC \setminus\{C\}$; this cycle must be contained in $G_C$, since otherwise it also shares edges with $C$. By the choice of $C$, $C''$ is in the set of marked cycles $\cC_M$. We claim that $C$ is then the lowest extension of $C''$, which again contradicts the assumption that $C \notin \cC_M$. Indeed, since both $C$ and~$C''$ are part of the solution $\cC$, and $C'' \le C$, we have that $C'' <_e C$. On the other hand, $C$ shares edges with $C'$, which is the only other maximal cycle in $G_C$ with respect to ($\le$). Thus, no other shortest cycle in $G_C$ is an extension of $C''$, and~$C = L(C'')$. This finishes the proof of the claim.
\end{claimproof}

For a set of cycles $\cC$ in $G$, let $\bigcup \cC$ be a shortcut for $\bigcup_{C \in \cC} C$, i.e., the graph obtained by taking the union of all cycles in $\cC$, which is a subgraph of $G$.
Let $H = \bigcup \cC_M$; in other words, $H$ is the subgraph of $G$ obtained by removing all vertices and edges that are not in any of the cycles in~$\cC_M$. Since $H$ is a subgraph of $G$, we treat the weight function $w$ as a weight function on edges of~$H$ as well.
From \Cref{claim:marked}, it follows immediately that $(G, w, k)$ is a yes-instance of \dsce if and only if $(H, w, k)$ is a yes-instance. Therefore, we focus on $(H, w, k)$ for the remainder of the proof.

Observe also that the LSCT $\cT(H)$ can be seen as the result of dissolving degree-$2$ nodes in~$\cT(G)$. Here, dissolving a degree-$2$ node means removing the node from the graph while making its neighbors adjacent. Indeed, by construction of $\cC_M$, the only shortest cycles in $G$ that are not necessarily in $\cC_M$ are $U$-nodes with exactly one child. Removing vertices and edges of such a cycle $C$ that are not part of~$\cC_M$ results in the child of $C$ being directly adjacent to the parent of $C$ in the new LSCT. Note that all leaves, $S$-nodes, $U$-nodes with at least two children and the root node are preserved between $\cT(G)$ and $\cT(H)$.

While $H$ is constructed from only $\OO(k^2)$ cycles of $G$, it is not necessarily a kernel yet, since the number of vertices involved in these cycles may be large.
We now argue that the number of vertices in~$H$ can be reduced by showing that only a few vertices in $H$ have degree more than $2$.
We identify the vertices of degree at least $3$ by following the LSCT $\cT(H)$.
For each $C \in \cT(H)$, we define the set of vertices $B(C) \subseteq V(H_C)$ as follows.
If $C$ is an $S$-node, then $B(C) = \{s, t\}$, where $s, t$ are the poles of $C$. If $C$ is a $U$-node, let $B(C)$ be the set of vertices of degree at least three in the graph $\bigcup \{C\} \cup \cC_u(C)$, i.e., we consider only the cycle $C$ and its children in $\cT(H)$.

\begin{claim}\label{claim:U_nodes_degree}
    For each $U$-node $C$ of $\cT(H)$ with $p$ children, the size of $B(C)$ is at most $p \cdot (p + 1)$.
\end{claim}
\begin{claimproof}
   The family $\{C\} \cup \cC_u(C)$ contains $p + 1$ cycles. By \Cref{lem:trivial}, every two of these cycles either touch or do not intersect. Therefore, for every two cycles $C_1$ and $C_2$ there are at most two vertices of degree at least $3$ in $C_1 \cup C_2$, which are the endpoints of the shared path $C_1 \cap C_2$. On the other hand, every vertex of degree at least three in  $\bigcup \cC_u(C) \cup \{C\}$ has two incident edges that belong to some distinct cycles $C_1, C_2 \in \{C\} \cup \cC_u(C)$, which means that this vertex has degree at least three also in~$C_1 \cup C_2$. Thus, the size of $B(C)$ is at most twice the number of pairs of cycles in $\{C\} \cup \cC_u(C)$, showing the claim.
\end{claimproof}

\begin{claim}\label{claim:high_degree}
    The number of vertices of degree at least three in $H$ is in $\OO(k^2)$.
\end{claim}
\begin{claimproof}
    Let $B = \bigcup_{C \in \cT(H)} B(C)$.
    Let $v$ be a vertex of degree at least $3$ in $H$, we claim that $B$ contains $v$. By definition, $B$ contains the poles of all splittable cycles, so it remains to consider the case where $v$ is not a pole of any splittable cycle.
    Since the degree of $v$ is at least three in $H$, there exist two distinct shortest cycles~$C_1$, $C_2$ in $H$ that contain distinct edges incident to $v$. Assume~$C_1$ is a non-tree cycle, then by \Cref{lem:tree}, (ii), there exist a splittable cycle $C \in \cT(H)$ with $\cP_s(C) = \{P_0, \ldots, P_\ell\}$ and~$C_1 = P_i \cup P_j$ for some $0 \le i < j \le \ell$. Since~$v$ is not a pole of $C$,~$v$ is an internal vertex of~$P_i$ or~$P_j$; in any case, there exists a cycle in~$\cC_s(C)$ that contains the same edges incident to~$v$ as~$C_1$, so~$C_1$ may be replaced by this cycle.
    Since an analogous argument can be applied to $C_2$, from now on we assume that both~$C_1$ and~$C_2$ are nodes of $\cT(H)$.

    Let $C_1$, $C_2$ be the closest pair of cycles with this property, where the distance is measured over the tree $\cT(H)$, and let $C_1$ be the highest node in $\cT(H)$ among $C_1$, $C_2$. We claim that either $C_1$ is the parent of $C_2$ in $\cT(H)$, or $C_1$ and $C_2$ have a common parent in $\cT(H)$. Assume the contrary, and consider two cases depending on the locations of $C_1$ and $C_2$ in $\cT(H)$.
    
    \noindent\textbf{Case 1} ($C_2 \le C_1$): Consider the parent $C$ of $C_2$ in $\cT(H)$, $C \le C_1$ and $C \ne C_1$. The cycle~$C$ contains~$v$, since otherwise $v$ is an inner vertex of $H_C$ and cannot be on the cycle $C_1$.
    If $C$ contains the same edges incident to $v$ as $C_2$, then $C_1$, $C$ fulfill the same property but are closer in $\cT(H)$ than $C_1$, $C_2$, which is a contradiction. Otherwise, $C$ contains an edge incident to $v$ that is not contained in $C_2$, and the pair~$(C, C_2)$ fulfills the property, leading to a contradiction.
    
    \noindent\textbf{Case 2} ($C_2 \nleq C_1$ and $C_1 \nleq C_2$): Let $C$ be the parent of $C_2$ in $\cT(H)$. By the assumption, $C$ is not the parent of $C_1$, and also $C_1 \nleq C$ since otherwise $C_1$ has lower depth in $\cT(H)$ than $C_2$. Similarly to the previous case, $C$ contains $v$, as otherwise $C_1$ and $C_2$ cannot share the vertex. Again, either $C$ contains the same edges incident to $v$ as $C_2$, or it contains an edge incident to~$v$ that is not part of $C_2$. Thus, either $C_1$ and $C$ or $C$ and $C_2$ fulfill the property while being closer in $\cT(H)$ than $C_1$ and $C_2$, which is a contradiction.

    We now have that either $C_1$ is the parent of $C_2$, or there is a cycle $C$ that is the parent of both $C_1$ and $C_2$ in $\cT(H)$.
    If $C_1$ is the parent of $C_2$, and $C_1$ is an $S$-node, then $v$ is a pole, which we assume is not the case; the same holds if $C$ is an $S$-node and $C_1, C_2 \in \cC_s(C)$. If $C_1$ is the parent of $C_2$ and $C_1$ is a $U$-node, then $C_1 \cup C_2$ is a subgraph of $\bigcup \cC_u(C_1) \cup \{C_1\}$, so $v \in B(C_1)$. By the same argument, if $C$ is a $U$-node and $C_1, C_2 \in \cC_u(C)$, then $v$ has degree at least three in~$C_1 \cup C_2$ and so in $\bigcup \cC_u(C) \cup \{C\}$, thus $v \in B(C)$.
    This concludes the proof that $B$ contains all vertices of degree at least $3$ in $H$.

    It remains to bound the number of elements in $B$. For each $C \in \cT(H)$ that is either an $S$-node or a $U$-node with one child, it holds that~$|B(C)| \le 2$.
    Since $|\cC_M| \in \OO(k^2)$, the total contribution of these nodes to $B$ is $\OO(k^2)$.
    Observe that $\cT(H)$ has at most $4k$ leaves, since the leaves of $\cT(H)$ and the leaves of $\cT(G)$ are the same.
    Therefore, the total number of children of all $U$-nodes with at least two children in $\cT(H)$ is at most $8k$. Since a $U$-node with $p$ children contributes at most $p \cdot (p + 1)$ vertices to $B$ by \Cref{claim:U_nodes_degree}, all $U$-nodes with at least two children contribute in total at most $8k \cdot (8k + 1) \in \OO(k^2)$ vertices to $B$. The total size of $B$ is therefore also $\OO(k^2)$.
\end{claimproof}

We now exhaustively apply the following reduction rules, in order to remove all degree-$2$ vertices in~$H$.
\begin{redrule}\label{redrule:dissolve}
If there is a vertex of degree $2$ in $H$, dissolve it: remove the vertex and add an edge between its former neighbors, such that the weight of the new edge is equal to the sum of weights of the two removed edges.
\end{redrule}
Note that \ref{redrule:dissolve} may create loops and parallel edges. Thus we allow for intermediate instances to be multigraphs. We introduce two more reduction rules so that the resulting instance is a simple graph.
\begin{redrule}\label{redrule:loops}
If there is a loop $e$ in $H$, remove it from the graph. If $w(e) = g(G)$, additionally decrease $k$ by $1$. If this increases the length of the shortest cycle, return a trivial no-instance unless $k$ becomes $0$, in which case return a trivial yes-instance.
\end{redrule}
\begin{redrule}\label{redrule:parallel}
If there are two parallel edges $e_1$ and $e_2$ in $H$, remove both and decrease~$k$ by $1$ if $w(e_1) + w(e_2) = g(G)$. If this increases the length of the shortest cycle, return a trivial no-instance unless $k$ becomes $0$, in which case return a trivial yes-instance. Otherwise, if~$w(e_1) + w(e_2) > g(G)$, then remove one of the parallel edges with the highest weight without changing $k$.
\end{redrule}

\begin{claim}\label{claim:redrules}
    \Cref{redrule:dissolve,redrule:loops,redrule:parallel} are safe.
\end{claim}
\begin{claimproof}
    First, for \Cref{redrule:dissolve}, observe that every cycle in $H$ either remains unchanged, or is replaced by a cycle that passes through the newly-created edge, while the length of the cycle and the set of remaining edges stays the same. Therefore, the length of all cycles is unchanged, any two cycles that were edge-disjoint stay edge-disjoint, and any two cycles that shared edges still share edges after applying the reduction rule.

    Second, consider a loop $e$ to which \Cref{redrule:loops} is applied. If $w(e) > g(G)$, no shortest cycle passes through~$e$, thus it can be safely removed. If $w(e) = g(G)$, the only shortest cycle passing through~$e$ is the one that contains the edge $e$ and nothing else, as only simple cycles are considered. Adding this cycle to the solution is safe as no other cycle in $H$ can share edges with it.

    Finally, let $e_1$, $e_2$ be the parallel edges in \Cref{redrule:parallel}, with the endpoints $u, v$. We start with the case $w(e_1) + w(e_2) = g(G)$. If $k = 1$, then the cycle $C$ containing $e_1$ and~$e_2$ is indeed a solution; if $k > 1$ and $g(H') > g(H)$, where $H' = H - e_1 - e_2$, then $(H, w, k)$ is a no-instance: Otherwise, each cycle in the solution must use a different edge in $\{e_1, e_2\}$, so $k = 2$, but then the remaining $u$-$v$-paths of the two cycles in $H'$ form a cycle of length~$g(G) = g(H)$ in~$H'$, contradicting $g(H) < g(H')$. Now, with $g(H') = g(H)$, we claim that a packing of $k$ edge-disjoint shortest cycles exist in~$H$ if and only if a packing of $k - 1$ edge-disjoint shortest cycles exists in~$H' = H$. The backward direction is clear, since the cycle~$C$ containing $e_1$ and~$e_2$ can be added to any packing in $H'$. In the forward direction, if the statement does not hold, then the solution in $H$ does not contain the cycle $C$, but contains at least one cycle with an edge in~$\{e_1, e_2\}$. If there is just one such cycle, then it can be replaced by $C$, so that the remaining~$k - 1$ cycles form a packing in $H$. If there are two cycles $C_1$ and $C_2$, containing $e_1$ and $e_2$ respectively, let $P_1$ be the $u$-$v$-path in $C_1$ other than $e_1$, and $P_2$ be the $u$-$v$-path in $C_2$ other than $e_2$. Since~$w(C_1) + w(C_2) = 2g(G)$ and~$w(e_1) + w(e_2) = g(G)$, it holds that~$w(P_1) + w(P_2) = g(G)$ and the cycle $C' = P_1 \cup P_2$ is also a shortest cycle.
    However, we can then replace $C_1$ and $C_2$ with~$C$ and $C'$ to get a packing of $k$ edge-disjoint shortest cycles in~$H$ that contains $C$.

    In the other case let without loss of generality be~$w(e_2) \geq w(e_1)$. If~${w(e_1) < w(e_2)}$, then no shortest cycle may pass through $e_2$, so removing the edge $e_2$ is safe. If~$w(e_1) = w(e_2) > g(G)/2$, then no shortest cycle can contain both $e_1$ and~$e_2$. Also, it cannot be that two edge-disjoint shortest cycles contain $e_1$ and~$e_2$, respectively, as then by joining~$C_1 - e_1$ and~$C_2 - e_2$ we get a cycle that is shorter than $g(G)$. Thus, in any solution at most one cycle passes through either $e_1$ or $e_2$. Thus, it is safe to remove $e_2$ since any such cycle may pass through $e_1$ instead.
\end{claimproof}

By applying \Cref{redrule:dissolve,redrule:loops,redrule:parallel} exhaustively, we obtain a simple graph $G'$ with the weight function~$w'$, and an integer $k' \le k$ such that $(G, w, k)$ is equivalent to $(G', w', k')$ (by \Cref{claim:redrules}). Moreover, $|V(G')| \in \OO(k^2)$, by \Cref{claim:high_degree} and since \Cref{redrule:dissolve} can no longer be applied. It remains to compress the weights, so that they can be represented by~$k^{\OO(1)}$ bits.
Let $n = |V(G')|$, $m = |E(G'|$, we have that $m \in \OO(k^4)$. We now interpret the weight function $w' : E(G') \to \mathbb{Z}$ as a vector in~$\mathbb{Z}^m$ by assigning an arbitrary numbering to the edges of $G'$, and invoke \Cref{prop:FT} with~$w'$ and $N = m + 1$ to obtain a compressed vector $\overline{w} \in \mathbb{Z}^m$. From \Cref{prop:FT}, we have that for any vector~$b \in \mathbb{Z}^m$ with~$||b||_1 \le m$, $\mathsf{sign}(w'\cdot b)=\mathsf{sign}(\overline{w}\cdot b)$. Interpreting the vector $\overline{w}$ as a weight function on $E(G')$, we obtain the instance $(G', \overline{w}, k')$ of \dsce. By \Cref{prop:FT}, ${||\overline{w}||_{\infty} \le 2^{4m^3} \cdot (m + 1)^{m (m + 2)} \in 2^{\OO(k^{12})}}$, therefore each entry of $\overline{w}$ can be represented by~$\OO(k^{12})$~bits. The bit size of~$(G', \overline{w}, k')$ is thus in~$\OO(k^{16})$.
It remains to verify that $(G', \overline{w}, k')$ is indeed equivalent to~$(G', w', k')$. We show that for every pair of cycles in $G'$, their weights compare in the same way before and after the weight compression.

\begin{claim}\label{claim:weight_compression}
    For every two cycles $C_1$, $C_2$ in $G'$, $w'(C_1)$ is smaller than (equal to/larger than)~$w'(C_2)$ if and only if $\overline{w}(C_1)$ is smaller than (equal to/larger than) $\overline{w}(C_2)$.
\end{claim}
\begin{claimproof}
    It suffices to show that $\mathsf{sign}(w'(C_1) - w'(C_2))=\mathsf{sign}(\overline{w}(C_1) - \overline{w}(C_2))$.
    For each cycle~$C_i$ with $i \in \{1, 2\}$, consider its characteristic vector~$a_i \in \{0, 1\}^m$, where the $j$-th coordinate is equal to~$1$ if and only if the $j$-th edge of $G'$ belongs to $C_i$. By definition, $w'(C_i) = w' \cdot a_i$ and $\overline{w}(C_i) = \overline{w} \cdot a_i$ for each~$i \in \{1, 2\}$. Let $b \in \{-1, 0, 1\}^m$ be the difference vector $a_1 - a_2$, then $w'(C_1) - w'(C_2) = w' \cdot b$ and~$\overline{w}(C_1) - \overline{w}(C_2) = \overline{w} \cdot b$. By construction, $||b||_1 \le m$, so \Cref{prop:FT} applies to the vector $b$, therefore
    $\mathsf{sign}(w'\cdot b)=\mathsf{sign}(\overline{w}\cdot b)$, which means that $\mathsf{sign}(w'(C_1) - w'(C_2))=\mathsf{sign}(\overline{w}(C_1) - \overline{w}(C_2))$.
\end{claimproof}

From \Cref{claim:weight_compression} we have that every shortest cycle in the instance $(G', \overline{w}, k')$ is a shortest cycle in $(G', w', k')$, and the other way around.
Therefore, the instances are equivalent with respect to the existence of a packing of $k$ edge-disjoint shortest cycles.
Since $(G', w', k')$ was previously shown to be equivalent to the original instance $(G, w, k)$ of \dsce, $(G', \overline{w}, k')$ is equivalent to $(G, w, k)$ as well, and, from the above, the bit-size of $(G', \overline{w}, k')$ is $\OO(k^{16})$. This completes the proof of the theorem.
\end{proof}

By using the small family of marked cycles constructed in the proof of the theorem, we immediately get the following FPT algorithm for \dsce on planar graphs.

\begin{corollary}\label{cor:ed-FPT}
\dsce can be solved in $k^{\OO(k)}\cdot n^{\OO(1)}$ time on planar graphs.
\end{corollary}

\begin{proof}
Given the instance $(G, w, k)$ of \dsce, construct the set of marked cycles $\cC_M$ as in the proof of \Cref{thm:mincycles-ed}. By \Cref{claim:marked} we have that whenever there exists a solution for $(G, w, k)$, there also exists one where all cycles are in $\cC_M$. By \Cref{claim:marked_size}, we also have that~$\cC_M = \OO(k^2)$. Therefore, in time $k^{\OO(k)}\cdot n^{\OO(1)}$ we can enumerate all $k$-tuples of cycles in $\cC_M$ and check for each tuple whether the cycles are edge-disjoint.
\end{proof}

Finally, combining the equivalence between \dsce and \mcp on planar graphs, and the two results above, we get analogous results for \mcp on planar graphs.

\begin{corollary}\label{cor:min-cut-FPT}
\mcp on planar graphs admits a polynomial kernel such that the output graph has $\OO(k^2)$ vertices.
Furthermore, the problem can be solved in $k^{\OO(k)}\cdot n^{\OO(1)}$ time on planar graphs.
\end{corollary}

\section{Kernelization lower bound for \dsc and \dsce}\label{sec:kern-lb}

Note that \Cref{thm:minsumxpalg} implies that \dsc and \dsce are in \classXP, but since the lower bound in~\Cref{thm:minsum-lb} does not apply for these problems, it remains open whether they are \classFPT or \classW1-hard on general graphs.
Here, we show that both problems do not admit polynomial kernels.
Formally, we prove the following.

\begin{restatable}{theorem}{thmnokern}\label{thm:nokern}
\dsc and \dsce parameterized by solution size~$k$ do not admit polynomial kernels unless $\classNP\subseteq\classCoNP/{\sf poly}$.
\end{restatable}

To this end, we first define \dsf, a version of \textsc{Disjoint Factors} where all of the factors of the solution need to be shortest factors.
Let~$\Sigma$ be a finite alphabet.
We use~$\Sigma^n$ to denote all words of length~$n$ over~$\Sigma$.
Given a word~$w = w_1w_2\ldots w_n$, we denote the $i$\textsuperscript{th} charater in~$w$ by~$w[i]$, that is, $w[i]=w_i$.
An~$x$-factor of a word~$w = w_1w_2\ldots w_n$ is a substring~$w_iw_{i+1}\ldots w_j$ for some~${i < j \in [n]}$ which starts and ends with~$x$, that is, $w_i = x = w_j$.
The length of a factor~$w_iw_{i+1}\ldots w_j$ is $j-i$ and a~$x$-factor is a \emph{shortest factor} if there is no other $x$-factor which has smaller length.
Two factors~$w_1 = w_{i_1}w_{i_1+1}\ldots,w_{j_1}$ and~$w_2 = w_{i_2}w_{i_2+1}\ldots,w_{j_1}$ intersect if~$\{i_1,i_1+1,\ldots,j_1\} \cap \{i_2,i_2+1,\ldots,j_2\} \neq \emptyset$, that is, there is an index in which~$w_1$ and~$w_2$ intersect.
Otherwise, the factors are disjoint.
We can now define \dsf{} formally.

\defproblema{\dsf}%
{A word~$w$ over a finite alphabet~$\Sigma$ with~$\bl \in \Sigma$.}%
{Decide whether there is a set of pairwise disjoint shortest factors including an~$x$-factor for each letter~$x \in \Sigma \setminus \{\bl\}$.}

Note that we allow for a special character \bl{} for which no factor needs to be picked.
Since adding two consecutive \bl{} to the start of the input word adds a shortest factor for~\bl, this does not change the computational complexity of \dsf.
We show that \dsf{} does not admit a polynomial kernel when parameterized by the alphabet size. We believe that this result may be interesting by itself.
We mention that our reduction is similar to the reduction that shows that \textsc{Disjoint Factors} does not admit a polynomial kernel \cite{BTY11}.

\begin{proposition}\label{prop:short-df}
    \dsf{} parameterized by~$|\Sigma|$ does not admit a polynomial kernel.
\end{proposition}

\begin{proof}
    We present an OR-cross-decomposition from \textsc{3,4-Sat}, a variant of \textsc{Sat} where each clause contains exactly 3 literals and each variable appears at most four times.
    This version is known to be NP-complete~\cite{Tov84}.
    We first present a polynomial-time many-one reduction from \textsc{3,4-Sat} to \dsf.
    Given a formula~$\phi$ in 3-CNF-Sat, our alphabet consists of one character for each of the~$3m$ positions in~$\phi$ and one character for each variable~$x_i$ in~$\phi$.
    We create the word
    $$w=123123123456456456\ldots(3m-2)(3m-1)(3m)(3m-2)(3m-1)(3m)\bl \bl w_{x_1} w_{x_2} \ldots w_{x_n},$$ where $$w_{x_i} = x_i p^i_1 \bl \bl p^i_1 p^i_2 \bl \bl p^i_2 p^i_3 \bl \bl p^i_3 p^i_4 \bl \bl p^i_4 x_i n^i_1 \bl \bl n^i_1 n^i_2 \bl \bl n^i_2 n^i_3 \bl \bl n^i_3 n^i_4 \bl \bl n^i_4 x_i.$$
    Therein, $p^i_j$ (respectively $n^i_j$) are the positions where~$x_i$ appears for the~$j$\textsuperscript{th} time positively (respectively negated) or~$\bl$ if~$x_i$ does not appear~$j$~times positively (negated).
    Note that the shortest factor for any of the position characters has length three and for any variable character, it has length~$17$.
    If the formula~$\phi$ is satisfiable, then we find a shortest factor for each character in~$w$ as follows.
    Let~$\beta$ be a satisfying assignment for~$\phi$.
    For each variable~$x_i$, if~$\beta(x_i)$ is false, then we select the first~$x_i$-factor in~$w_{x_i}$ (the one through all~$p^i_j$'s).
    If~$\beta(x_i)$ is true, then we select the second~$x_i$-factor in~$w_{x_i}$ (the one through all~$n^i_j$'s).
    For each clause~$C_j$, we select one variable~$v_j$ such that~$\beta(v_j)$ satisfies~$C_j$.
    Note that such a variable exists as~$\beta$ is a satisfying assignment.
    Let~$c_j$ be the position of~$v_j$ in~$C_j$.
    Note that by construction, we can pick a shortest~$c_j$-factor in~$w_{v_j}$.
    For the other two positions in~$C_j$, we can find disjoint shortest factors in the starting part of~$w$.

    In the other direction, if~$w$ admits a set of disjoint shortest factors, then we will construct a satisfying assignment for~$\phi$.
    For each variable~$x_i$, $w$ only contains two shortest~$x_i$-factors and both are contained in~$w_{x_i}$.
    If the first one is chosen, then we set~$x_i$ to false and otherwise, we set~$x_i$ to true.
    Suppose this assignment does not satisfy~$\phi$.
    Then, there exists a clause~$C_j$ that is not satisfied.
    Note that for each position character, there are only 3 shortest factors in~$w$, two in the starting part and one in a word~$w_{x_i}$ for the variable in the respective position.
    Moreover in the starting part, we can find disjoint shortest factors for at most two out of the three positions~$(3j-2)$, $(3j-1)$, and~$3j$ in~$C_j$.
    This means that there is at least one position in~$C_j$ whose position-character-factor has been chosen within a variable word~$w_i$.
    This is a contradiction to the fact that the chosen assignment does not satisfy~$C_j$.
    
    We next present the OR-cross-decomposition.
    Given~$t$ instances of \textsc{3,4-Sat}, where all instances contain the same number of clauses and variables, we use the previously described many-one reduction to construct words~$w_1,w_2,\ldots,w_t$, one for each instance.
    We assume that~$t=2^s$ for some integer~$s$.
    Otherwise, we can copy one of the instances at most~$t$ times to ensure that the number of instances is a power of~$2$.
    Note that by construction, all words~$w_i$ have the same length~${N = 9m + 35n \in O(n+m)}$, all words are over the same alphabet~$\Sigma$, and the shortest factor for each character has the same length in each~$w_i$.
    Next, we combine all of these instances into one instance such that the whole instance is a yes-instance if and only if at least one of the words~$w_i$ is a yes-instance using~$s$ additional characters~$c_1,c_2,\ldots,c_{s}$.
    Let~$\Sigma' = \Sigma \cup \{c_1,c_2,\ldots,c_s\}$.
    We iteratively combine two batches of instances of the same size into one instance, that is, in the~$i$\textsuperscript{th} round, we combine two batches containing~$2^{i-1}$~instances into one batch containing~$2^i$ instances using character~$c_i$.
    Note that after~$s = \log_2(t)$ rounds, we get one instance encoding all~$t$ instances.

    We next describe how two batches are merged.
    To this end, let~$W_1$ and~$W_2$ be the words corresponding to two batches and let~$N_i = |W_1| = |W_2|$.
    We create the word $$W = c_i W_1 \bl^{\nicefrac{N_i}{2}} c_i \bl^{\nicefrac{N_i}{2}} W_2 c_i \bl.$$
    Note that~$N_i = 3N_{i-1} + 4$, where~$N_0 = N$.
    This implies that~$N_i = 3^iN + 4i$.
    In particular, ${N_{s} = 3^s N + 4s = t^{\log_2(3)}N + 4\log_2(t)}$, that is, the whole constructed instances has size in~$\poly(tN)$.
    Moreover,~$|\Sigma'| \leq N + \log_2(t)$.
    It hence only remains to show that the constructed instance is a yes-instance if and only if at least one~$w_i$ is a yes-instance.
    To this end, first observe that all shortest~$c_i$-factors are contained in the words for batches of~$2^{i}$ input words as we introduced~$N_i$ \bl{} between two such words when creating batches for~$i+1$.
    
    Assume first that one word~$w_i$ is a yes-instance.
    Then we find a set of disjoint shortest factors for all characters in~$\Sigma'$ in~$W$ as follows.
    For each character~$a \in \Sigma$, we take a shortest~$a$-factor from~$w_i$.
    Next for each~$c_i$, we look at the batch word~$W_i$ in round~$i$ that contains~$w_i$.
    Note that~$W_i$ contains the character~$c_i$ three times and between any two of them, there is a shortest~$c_i$-factor.
    We pick the one that does not intersect with~$w_i$.
    Note that in this way we can find a set of disjoint shortest factors for all characters in~$\Sigma'$.

    Now assume that the word~$W$ is a yes-instance and let~$S$ be a solution.
    Similar to before, the character~$c_{\log_2(t)}$ is contained only three times in~$W$.
    We take the shortest~$c_{\log_2(t)}$-factor in~$S$ and consider the rest~$W'$ of the word~$W$.
    We iteratively select the shortest~$c_i$-factor in~$S$ for decreasing values of~$i$.
    After~$\log_2(t)$ iterations, we are only left with a single input word~$w_i$.
    By assumption, $S$ contains disjoint shortest factors for each character in~$\Sigma$, that is, the instance~$w_i$ is a yes-instance.
    This concludes the proof.
\end{proof}

Now we can prove~\Cref{thm:nokern} which we restate.

\thmnokern*

\begin{proof}
    We present a polynomial parameter transformation from \dsf.
    To this end, let~$w \in \Sigma^n$ be a string.
    We assume without loss of generality that~$n \geq 10$.
    We start with a path of length~$2n$.
    Let the vertices be~$u_1,u_2,\ldots,u_{2n}$.
    Next, for each~$a \in \Sigma$, we add a new vertex~$v_a$ to the graph.
    Let~$i^a_1 < i^a_2 < \ldots < i^a_j$ be all the positions~$i$ where~$w[i] = a$ and let~$d_a$ be the size of a shortest~$a$-factor in~$w$.
    We add a path between~$v_a$ and~$u_{2i_p}$ for each~$p \in [j]$.
    The path between~$v_a$ and~$u_{i_p}$ is of length~$3n - d_a$.
    Finally, we set~$k = |\Sigma|$.

    Since the reduction can be computed in polynomial time and \dsf{} does not admit a polynomial kernel when parameterized by~$|\Sigma|$, it only remains to show that~$w$ has~$k$ disjoint shortest factors if and only if there exists a set of~$k$ disjoint cycles of length~$g$ in the constructed graph, where~$g$ is the girth of the graph.
    We start by showing that the girth of the graph is~$6n$.
    Note that the set of vertices~$\{v_a \mid a \in \Sigma\}$ is a feedback vertex set of the constructed graph and hence each cycle passes through at least one of these vertices.
    Moreover, each cycle that passes through at least two such vertices are of length at least~$4(3n-n) = 8n$ as~$d_a \leq n$ for each~$a \in \Sigma$.
    Next, each cycle that passes through exactly one vertex~$v_a \in \{v_a \mid a \in \Sigma\}$ has length at least~$2(3n - d_a) + 2d_a = 6n$ as any two vertices of the initial path that are connected to the same vertex~$v_a$ have distance at least~$2d_a$ as any factor of~$w$ has length at least~$d_a$.
    Moreover, each cycle through~$v_a$ that has length exactly~$g = 6n$ corresponds to a shortest~$a$-factor in~$w$.
\end{proof}

We mention that forcing the shortest distances for all factors to be equal is basically the same as \textsc{Colorful Independent Set} in unit interval graphs where all monochromatic cliques have size at most~$2$.
This admits a cubic kernel as shown by \citet{BMNW15}.
For the sake of completeness, we sketch a cubic kernel directly but based on their ideas.
Observe that the distance~$\ell$ of any shortest factor is at most~$k$ as otherwise, any shortest factor would contain some letter twice, a contradiction to the factor being shortest and all shortest factors having the same length~$\ell$.
Next, we can add~$\ell+3$ new characters and build a quadratic-size gadget in the beginning that allows all of these to be picked in a solution. Afterwards, we can use these characters to replace deleted characters (since we cannot use a single character \bl{} in this variant).
The gadget looks as follows: $$1,3,4,\ldots,\ell+2,1,2,4,5,\ldots,\ell+3,3,1,5,6,\ldots,\ell+3,3,\ldots,\ell+3,1,2, \ldots, \ell,\ell+3.$$
With this at hand, we can observe that if there are at least~$2k-1$ shortest factors of a given letter, then each of these can intersect with at most two factors in any solution, that is, we can ignore this character (replace it by the new characters).
Now, each of the~$O(k)$ character has at most~$2k$ shortest factors, each of length at most~$k$.
Keeping only characters contained in these factors and adding ``separators'' of length~$k$ between ``connected components'' results in a cubic kernel.

\section{Packing Minimum Cuts}\label{sec:min-cuts}
In this section, we show that packing~$k$ minimum cuts in a graph is~W[1]-hard when parameterized by~$k$.
This is partially related to \dsc{} as packing $k$ minimum cuts in the dual of a planar graph is the same as as packing $k$ shortest cycles in the primal graph.

\begin{restatable}{theorem}{thmcuts}\label{thm:min-cuts-hard}
 \mcp is \classW{1}-hard when parameterized by solution size~$k$ on graphs with unit edge weights.    
\end{restatable}

\begin{proof}
    We present a reduction from \textsc{Independent Set} parameterized by solution size~$k$. Recall that the problem is \classW{1}-complete~\cite{CyganFKLMPPS15,DowneyF13}.
    To this end, let~${(G=(V,E),k)}$ be an instance of \textsc{Independent Set}.
    We create an equivalent instance~$(H=(V',E'),k)$ of \mcp as follows.
    Let~$\Delta$ be the maximum degree in~$G$.
    We set~$V' = V \cup U$ where~$U$ is a set of~$2\Delta+3$ new vertices.
    For the edge set~$E'$, we start with the edge set~$E \cup \{\{u,v\} \mid u,v \in U\}$, that is, we make~$U$ into a clique of size~$2\Delta+3$.
    Moreover for each vertex~$v \in V$, we add edges between~$v$ and~$2\Delta+1 - \deg_G(v)$ arbitrary vertices in~$U$.
    This concludes the construction.

    Since the solution sizes for both instances are the same and since the reduction can be computed in polynomial time, it only remains to show that~$G$ contains an independent set of size~$k$ if and only if there are~$k$ disjoint minimum cuts in~$H$.
    We first show that all minimum cuts in~$H$ separate exactly one vertex~$v \in V$ from the rest of the graph.
    Note that each such cut has size~$2 \Delta + 1$.
    Moreover, no cut of size at most~$2\Delta+1$ can separate two vertices in~$U$ as each proper cut in a clique of size~$\ell = 2\Delta+3$ contains at least~$\ell - 1 = 2\Delta+2$ edges.
    Lastly, consider a cut that separates two vertices~$u,v\in V$ from all vertices in~$U$.
    Such a cut contains at least~$2\Delta+1 -  \Delta = \Delta + 1$ edges between~$u$ and vertices in~$U$ and also at least~$\Delta + 1$ edges between~$v$ and vertices in~$U$.
    Since these two sets of edges are disjoint, such a cut has size at least~$2\Delta+2$ and is therefore not a minimum cut.

    This basically concludes the proof as each minimum cut in~$H$ now corresponds to selecting a single vertex~$v$ in~$G$ and two cuts are disjoint if and only if the two selected vertices are non-adjacent in~$G$, that is, a packing of~$k$ disjoint minimum cuts in~$H$ corresponds to an independent set of size~$k$ in~$G$.
\end{proof}

\section{Conclusion}
We investigated the parameterized complexity of finding packings of vertex or edge-disjoint cycles of bounded total length. In \Cref{thm:minsumxpalg}, we showed that \msc and \msedc are in \classXP when parameterized by the number of cycles.  
The result is tight in the sense that both problems are proven to be \classW{1}-hard in~\Cref{thm:minsum-lb}. The interesting special cases where we pack shortest cycles are \dsc and \dsce. Trivially, 
\Cref{thm:minsumxpalg} implies that \dsc and \dsce are in \classXP. However, our lower bound in~\Cref{thm:minsum-lb} does not apply here. This leads to an intriguing open problem as to whether \dsc and \dsce are \classFPT or \classW1-hard on general graphs. Currently, we can only rule out the existence of polynomial kernels (\Cref{thm:nokern}).
We also observed that \mcp---which is dual for \dsce on planar graph---becomes \classW{1}-hard for general case (\Cref{thm:min-cuts-hard}).

Further, we considered the case of planar graphs. We proved that \dsc and \dsce are \classFPT when parameterized by the number of cycles and, furthermore, \dsce admits a polynomial kernel. The most interesting open question here is whether \msc and \msedc are \classFPT on planar graphs. Our algorithms for \dsc and \dsce crucially depend on the structure of shortest cycles in planar graphs, and we cannot apply the same ideas for the more general length constraint. Further, does 
\dsc admit a polynomial kernel in the planar case? The kernelization algorithm for \dsce in~\Cref{thm:mincycles-ed} exploits the fact that a plane graph with at least $4k$ shortest facial cycles always has a packing of $k$ edge-disjoint shortest cycles by the four-color theorem, and we do not have similar properties in the case of vertex-disjoint cycles. On the other side, the reduction in~\Cref{thm:nokern} is 
far from planar and cannot be used to infer a similar  kernelization lower bound for {\dsc}.

Another interesting open problem for planar graphs is about packings of cycles from uncrossable families (we refer to~\cite{GoemansW98,schlomberg2023packing,schlomberg2024improved} for the definitions). Is it possible to extend our results for shortest cycles for these more general cycle families?  

Finally, we would be interested to know whether the running time of our algorithms could be improved. The algorithm from \Cref{thm:minsumxpalg} runs in $n^{\OO(k^6)}$ time. Can the running time be improved to, say, $n^{\OO(k)}$? For the planar case, we solve \dsc and \dsce in $k^{\OO(k)}\cdot n^{\OO(1)}$ time. Is there a single-exponential in $k$ algorithm? It may even be that that these problems can solved in subexponential time. 

\newpage

\bibliographystyle{plainnat}
\bibliography{biblio,ref}

\newpage

\end{document}

%% file: Tree.pdf_t
\begin{picture}(0,0)%
\includegraphics{Tree.pdf}%
\end{picture}%
\setlength{\unitlength}{3947sp}%
\begin{picture}(7155,2391)(511,-2002)
\put(6751,239){\makebox(0,0)[lb]{\smash{\fontsize{12}{14.4}\usefont{T1}{ptm}{m}{n}{\color[rgb]{0,0,0}Root}%
}}}
\put(976,-886){\makebox(0,0)[lb]{\smash{\fontsize{12}{14.4}\usefont{T1}{ptm}{m}{n}{\color[rgb]{0,0,0}$C_3$}%
}}}
\put(1876,-886){\makebox(0,0)[lb]{\smash{\fontsize{12}{14.4}\usefont{T1}{ptm}{m}{n}{\color[rgb]{0,0,0}$C_4$}%
}}}
\put(3826,-886){\makebox(0,0)[lb]{\smash{\fontsize{12}{14.4}\usefont{T1}{ptm}{m}{n}{\color[rgb]{0,0,0}$C_5$}%
}}}
\put(526, 14){\makebox(0,0)[lb]{\smash{\fontsize{12}{14.4}\usefont{T1}{ptm}{m}{n}{\color[rgb]{0,0,0}$C_1$}%
}}}
\put(3001, 14){\makebox(0,0)[lb]{\smash{\fontsize{12}{14.4}\usefont{T1}{ptm}{m}{n}{\color[rgb]{0,0,0}$C_2$}%
}}}
\put(5851,-886){\makebox(0,0)[lb]{\smash{\fontsize{12}{14.4}\usefont{T1}{ptm}{m}{n}{\color[rgb]{0,0,0}$C_1$}%
}}}
\put(7651,-886){\makebox(0,0)[lb]{\smash{\fontsize{12}{14.4}\usefont{T1}{ptm}{m}{n}{\color[rgb]{0,0,0}$C_2$}%
}}}
\put(5251,-1936){\makebox(0,0)[lb]{\smash{\fontsize{12}{14.4}\usefont{T1}{ptm}{m}{n}{\color[rgb]{0,0,0}$C_3$}%
}}}
\put(6451,-1936){\makebox(0,0)[lb]{\smash{\fontsize{12}{14.4}\usefont{T1}{ptm}{m}{n}{\color[rgb]{0,0,0}$C_4$}%
}}}
\put(7651,-1936){\makebox(0,0)[lb]{\smash{\fontsize{12}{14.4}\usefont{T1}{ptm}{m}{n}{\color[rgb]{0,0,0}$C_5$}%
}}}
\end{picture}%

%% file: Decomp.pdf_t
\begin{picture}(0,0)%
\includegraphics{Decomp.pdf}%
\end{picture}%
\setlength{\unitlength}{3947sp}%
\begingroup\makeatletter\ifx\SetFigFont\undefined%
\gdef\SetFigFont#1#2#3#4#5{%
  \reset@font\fontsize{#1}{#2pt}%
  \fontfamily{#3}\fontseries{#4}\fontshape{#5}%
  \selectfont}%
\fi\endgroup%
\begin{picture}(4576,2680)(398,-2459)
\put(3976,-2386){\makebox(0,0)[lb]{\smash{{\SetFigFont{12}{14.4}{\rmdefault}{\mddefault}{\updefault}{\color[rgb]{0,0,0}b)}%
}}}}
\put(413,-439){\makebox(0,0)[lb]{\smash{{\SetFigFont{12}{14.4}{\rmdefault}{\mddefault}{\updefault}{\color[rgb]{0,0,0}$P_0$}%
}}}}
\put(754,-806){\makebox(0,0)[lb]{\smash{{\SetFigFont{12}{14.4}{\rmdefault}{\mddefault}{\updefault}{\color[rgb]{0,0,0}$P_1$}%
}}}}
\put(2401,-579){\makebox(0,0)[lb]{\smash{{\SetFigFont{12}{14.4}{\rmdefault}{\mddefault}{\updefault}{\color[rgb]{0,0,0}$P_\ell$}%
}}}}
\put(674,-1152){\makebox(0,0)[lb]{\smash{{\SetFigFont{12}{14.4}{\rmdefault}{\mddefault}{\updefault}{\color[rgb]{0,0,0}$C_1$}%
}}}}
\put(2127,-926){\makebox(0,0)[lb]{\smash{{\SetFigFont{12}{14.4}{\rmdefault}{\mddefault}{\updefault}{\color[rgb]{0,0,0}$C_\ell$}%
}}}}
\put(1453,-2132){\makebox(0,0)[lb]{\smash{{\SetFigFont{12}{14.4}{\rmdefault}{\mddefault}{\updefault}{\color[rgb]{0,0,0}$s$}%
}}}}
\put(1473, 62){\makebox(0,0)[lb]{\smash{{\SetFigFont{12}{14.4}{\rmdefault}{\mddefault}{\updefault}{\color[rgb]{0,0,0}$t$}%
}}}}
\put(601,-1786){\makebox(0,0)[lb]{\smash{{\SetFigFont{12}{14.4}{\rmdefault}{\mddefault}{\updefault}{\color[rgb]{0,0,0}$C$}%
}}}}
\put(3151,-1786){\makebox(0,0)[lb]{\smash{{\SetFigFont{12}{14.4}{\rmdefault}{\mddefault}{\updefault}{\color[rgb]{0,0,0}$C$}%
}}}}
\put(3347,-1259){\makebox(0,0)[lb]{\smash{{\SetFigFont{12}{14.4}{\rmdefault}{\mddefault}{\updefault}{\color[rgb]{0,0,0}$C_1$}%
}}}}
\put(4034,-1519){\makebox(0,0)[lb]{\smash{{\SetFigFont{12}{14.4}{\rmdefault}{\mddefault}{\updefault}{\color[rgb]{0,0,0}$C_2$}%
}}}}
\put(4461,-979){\makebox(0,0)[lb]{\smash{{\SetFigFont{12}{14.4}{\rmdefault}{\mddefault}{\updefault}{\color[rgb]{0,0,0}$C_3$}%
}}}}
\put(3827,-825){\makebox(0,0)[lb]{\smash{{\SetFigFont{12}{14.4}{\rmdefault}{\mddefault}{\updefault}{\color[rgb]{0,0,0}$C_4$}%
}}}}
\put(1426,-2386){\makebox(0,0)[lb]{\smash{{\SetFigFont{12}{14.4}{\rmdefault}{\mddefault}{\updefault}{\color[rgb]{0,0,0}a)}%
}}}}
\end{picture}%